\documentclass[authoryear,11pt]{elsarticle}
\usepackage{graphicx} 
\usepackage{fullpage}
\usepackage{amsmath}
\usepackage{amsthm}

\usepackage{setspace}
\usepackage{tikz}
\usetikzlibrary{positioning, arrows.meta, shapes.geometric, calc, backgrounds}

\usepackage{cleveref}
\usepackage{booktabs}
\usepackage{url}
\usepackage{xspace}
\usepackage{thm-restate}
\usepackage{enumitem}
\newtheorem{theorem}{Theorem}
\newtheorem{observation}[theorem]{Observation}
\newtheorem{corollary}[theorem]{Corollary}
\newtheorem{conjecture}[theorem]{Conjecture}
\newtheorem{lemma}[theorem]{Lemma}
\newtheorem{definition}[theorem]{Definition}
\newtheorem{example}[theorem]{Example}

\theoremstyle{definition}
\newtheorem{construction}{Construction}
\Crefname{construction}{Construction}{Constructions}

\usepackage{soul}

\usepackage{amssymb}

\newcommand{\W}[1]{{\sf W[#1]}\xspace}
\newcommand{\FPT}{{\sf FPT}\xspace}
\newcommand{\XP}{{\sf XP}\xspace}
\newcommand{\pNP}{{\sf para-NP}\xspace}
\newcommand{\XNLP}{{\sf XNLP}\xspace}
\newcommand{\LOG}{{\sf L}\xspace}

\newcommand{\FNC}{{\sf FNC}\xspace}
\makeatletter
\providecommand*{\shuffle}{%
  \mathbin{\mathpalette\shuffle@{}}%
}
\newcommand*{\shuffle@}[2]{%
  \sbox0{$#1\vcenter{}$}%
  \kern .15\ht0 
  \rlap{\vrule height .25\ht0 depth 0pt width 2.5\ht0}%
  \raise.1\ht0\hbox to 2.5\ht0{%
    \vrule height 1.75\ht0 depth -.1\ht0 width .17\ht0 %
    \hfill
    \vrule height 1.75\ht0 depth -.1\ht0 width .17\ht0 %
    \hfill
    \vrule height 1.75\ht0 depth -.1\ht0 width .17\ht0 %
  }%
  \kern .15\ht0 
}
\makeatother

\newtheorem{claim}{Claim}
\Crefname{claim}{Claim}{Claims}

\makeatletter
\def\ps@pprintTitle{%
  \let\@oddhead\@empty
  \let\@evenhead\@empty
  \let\@oddfoot\@empty
  \let\@evenfoot\@oddfoot
}
\makeatother
\begin{document}

\begin{frontmatter}

\title{Parameterized Complexity of Scheduling Problems in Robotic Process Automation}

\author[1]{Michal Dvořák}
\ead{michal.dvorak@fit.cvut.cz}

\author[2]{Antonín Novák}
\ead{antonin.novak@cvut.cz}

\author[2]{Přemysl Šůcha\corref{cor1}}
\ead{suchap@cvut.cz}

\author[1]{Dušan Knop}
\ead{dusan.knop@fit.cvut.cz}

\author[3]{Claire Hanen}
\ead{claire.hanen@lip6.fr}

\cortext[cor1]{Corresponding author}

\affiliation[1]{organization={Faculty of Information Technology, Czech Technical University in Prague},city={Prague},country={Czech Republic}}
\affiliation[2]{organization={Czech Institute of Informatics, Robotics and Cybernetics, Czech Technical University in Prague},
                   city={Prague},
                   country={Czech Republic}}

\affiliation[3]{organization={CNRS, LIP6, F-75005, Sorbonne University},
                   city={Paris},
                   country={France}}
\date{\today}

\begin{abstract}
    This paper studies the growing domain of Robotic Process Automation (RPA) problems. Motivated by scheduling problems arising in RPA, we study the parameterized complexity of the single-machine problem with precedence constraints, release times, and deadlines (i.e., the problem known as $1|\operatorname{prec},r_j,d_j|*$ in the three-field notation). We focus on parameters naturally linked to RPA systems, including chain-like precedences, the number of distinct processing times, and the structure of the time windows. We show that the problem is strongly \XNLP-hard parameterized by the number of chains, even with only two prescribed processing times and two distinct time-window lengths. The problem remains \XNLP-hard even under prec-consistent time windows. On the positive side, we obtain polynomial-time algorithm when all jobs share a single time-window length and \FPT when the processing times, release times and deadlines are chain-uniform. We also show that the problem lies in \XNLP when parameterized by the width of the precedence relation either when the instance is encoded in unary or the processing times are bounded.

\end{abstract}

\begin{keyword}
Scheduling \sep Precedences \sep Robotic Process Automation \sep Parameterized Complexity 
\end{keyword}

\end{frontmatter}


\clearpage
\section{Introduction}

With the increasing market competition, companies are pressured to seek ways to remain competitive while increasing their revenue, decreasing costs, and streamlining their operations.
Achieving those goals is often connected with improving the efficiency of companies' business processes.
To run daily operations, large companies developed internal processes for business-related activities such as invoice processing, payment handling, and customer relationship management, but also for their back-office operations such as report generation and human resources activities, e.g., employee onboarding, payroll, etc.
These workflows used to be operated by humans who are susceptible to errors, do not work 24/7, and are costly to scale up.
Hence, to enable further growth and increase efficiency, companies are forced to automate their processes.

Robotic Process Automation (RPA)~\citep{wewerka2020robotic} is quickly becoming the major technology for automating repetitive, rule-based tasks across many sectors by emulating human interactions with software, thus delivering faster throughput, higher accuracy, and continuous operation without risks of human fatigue. 
The increasing adoption of RPA systems by large enterprises has attracted the research community, which started to address different aspects of RPA---economic perspective~\citep{aguirre2017automation,lamberton2017impact}, reports on real case-studies~\citep{huang2019applying} or mathematical optimization~\citep{seguin2021minimizing}.

A typical RPA system consists of a set of business processes that are automated.
Each process performs operations on the individual data entries, also called \textit{items}, which represent the basic unit of work.
An example of a process may be a data extraction from invoices sent by suppliers via email, with the individual operation representing an optical character recognition over the PDF file and storing the extracted entries in the database.
The execution of the processes might be constrained by time windows imposed by the company's internal rules, customer requirements, or technological limitations.
Each operation performed on a data entry is carried out by a \textit{software robot}, which requires a hardware resource to run on (i.e., a machine, in scheduling terminology), and may require a software license.
The number of licenses is typically limited~\citep{seguin2021minimizing} as the companies are charged by RPA vendors based on the number of robot licenses used.
Items to be handled by a process are stored in a given number $k\geq 0$ of \textit{queues} and the process retrieves the items (typically, but not exclusively) in the first-in first-out (FIFO) fashion.
Hence, the order in which items are processed is constrained by chain precedence, or a general directed acyclic graph with limited width $w$. 
The processing time of an item $j$ is denoted by $p_j$.
Generally, they may be arbitrary, but often they are dependent on the queue in which they appear.

The RPA system consists of a set of $k$ queues that contain items to be processed.
A process accesses a queue $Q$ and spends $p_j$ time units of work to process item~$j$.
We recognize two main types of queues: (i)~the ones with uniform workload, meaning that all items $j$ in queue $Q$ have equal processing time $p_j = p_Q$ (i.e., transactions of the same type), and (ii)~with non-uniform workload, in which every item may have a different processing time (i.e., different amount of work to process an item is needed).

Item $j$ enters a specific queue at time $r_j$, corresponding to the release time.  
Deadline $d_j$ on the completion time of item $j$ is imposed by the existence of Service Level Agreement (SLA) $\Delta$ as $d_j=r_j+\Delta$, for some value of $\Delta$. Value $\Delta$ is typically a constant for all items or for items of a specific queue $Q$. 
The processes use queues to communicate with each other, i.e., passing the items between them, via a publisher-subscriber architecture.
This means a process can either put or consume items from a given queue.
Most typically, we encounter two types of processes---\textit{loaders} and \textit{worker} processes.
A loader process puts new items to be processed (i.e., generating the workload for another process) into the queue at a time $t$ in batches, meaning that all items in the batch share the same release time $r_j=t$.
If at time $t$ only items from a single batch are present in a queue $Q$, then all items share the common release time $r_Q$.
As a consequence of the existence of SLA $\Delta$, all deadlines $d_j$ are equal to a constant ${d}_Q$ as well.
Further, a worker process subscribes to the queue and processes the items in the FIFO fashion.
Hence, the processing of items by a worker process is constrained by chain precedences. In the context of RPA systems, what are conventionally termed items will henceforth be referred to as \emph{jobs}, aligning our terminology with that commonly employed in the scheduling domain.

For this purpose we adopt the standard three-field notation $\alpha|\beta|\gamma$, introduced by~\cite{GRAHAM1979287}. In the three-field notation, $\alpha$ represents the scheduling environment, which in this paper is the number of machines: one machine ($1$), or $m$ unrelated parallel machines ($Pm$), or unspecified number of unrelated parallel machines ($P$). The $\beta$ field represents the job characteristics. These can be for example: precedences ($\operatorname{prec}$), chains ($\operatorname{chains}$, release times of jobs ($r_j$), deadlines of jobs ($d_j$), specifications of processing times ($p_j=1$,$\ldots$). Finally, the $\gamma$ field represents the objective of the problem. In this paper, we consider either minimizing the makespan $C_{\max}$, which is the maximum completion time of a job, or an $*$ indicating that we are only interested in feasible schedule.

For efficient scheduling in these complex systems, it is insufficient to rely just on simple online rules governing the execution of activities.
Hence, we would like to retrieve the schedule over a rolling horizon as the result of a well-defined optimization problem.
To design efficient algorithms to solve these problems, we need to understand the complexity of its various aspects---namely, its (potential) tractability if some of the natural parameters such as number of queues $k$, precedence width $w$, maximum processing time $p_{\max}$, number of distinct processing times $\#p_j$, or number of distinct SLA's (i.e., number of distinct time windows) $\#(d_j-r_j)$ (see \Cref{sec:preliminaries} for formal definitions) are bounded by a constant which is fairly small.
This finer complexity characterization, also known as \textit{parameterized complexity}, may provide the needed problem insights.
It is known that in scheduling environments with parallel machines, already the simplest problems tend to be hard. For example, the problem $Pm|r_j,d_j|*$ is weakly \textsf{NP}-hard for $m=2$ machines and strongly \W{1}-hard parameterized by $m$~\citep{vanBevernRolfSuchy2017}. Similarly, $P2|\operatorname{chains}(k),r_j,d_j|C_{\max}$ is {\sf NP}-hard even if $k$ and path-width of the underlying interval graph of the time windows is constant~\citep{Hanen2024ParadependentTasks} (see \Cref{sec:preliminaries} for definition of path-width). Even without release times and deadlines, the problem is {\sf NP}-hard for $m=2$ machines and $k=3$ chains.

Therefore, in studying parameterized complexity of the problems that result from RPA scheduling, we rather focus on a family of single-machine problems with respect to the key parameters of RPAs as identified in Table~\ref{tab:table_results}. However, in
\Cref{sec:rpa_pc_relation_discussion}, we summarize the complexity results and formulate strategies for designing RPA scheduling involving parallel machine settings.

\section{Related work}\label{sec:related_work}

\subsection*{Robotic process automation}
Most of the work on RPA focuses on economic studies quantifying the effects of implementing robotic automation in various industries \citep{ribeiro2021robotic,prucha2024roboticprocessautomationdriver}, but operations research literature remains largely underdeveloped.
A notable exception consists of the work by Séguin et al.~\citep{seguin2020robotic,seguin2021minimizing}, which developed a simplified mathematical model of robots processing business transactions.
They coin the problem of minimizing the cost of robot licenses as an optimization problem solved by a mixed-integer linear programming model.
They describe the whole range of relevant aspects of RPA scheduling, for example, they deal with different process types, the existence of licenses, deadlines, etc.
However, their work focuses more on capacity planning than the scheduling problem, disregarding the sequencing aspect of the problem.
As a result, their formulation leads to a bin-packing type of optimization problem.
For example, the precedence relations between the items are not considered, which becomes an increasingly relevant feature of RPA systems, as the complexity of the automated workflows grows rapidly.

\subsection*{Scheduling}
Scheduling with release times, processing times, and deadlines corresponds to the classical problem \textsc{Sequencing With Release Times And Deadlines}, i.e., $1|r_j,d_j|*$ which is known to be strongly \textsf{NP}-hard~\cite[p.~236, SS1]{GareyJ1979}. 

\cite{Lawler1973} then studied the scheduling problems under precedence constraints and showed that the problem $1|\operatorname{prec},r_j|C_{\max}$ is solvable in polynomial time. In the same paper, \cite{Lawler1973} showed that the problem $1|\operatorname{prec}|\max_j f_j(C_j)$ problem can be solved in quadratic time where the functions $f_j$ are any non-decreasing cost functions of the completion times.
\cite{Horn72Treelike}~and~\cite{Sidney75} independently developed an $O(n\log n)$ algorithm for $1|\operatorname{tree}|\sum w_j C_j$, that is, when the precedence constraints are restricted to be a tree. Later, an efficient algorithm for more generalized series-parallel constraints was developed by~\cite{Adolphson77}. From the negative side, \cite{Lenstra78} showed that in general $1|\operatorname{prec}|\sum w_j C_j$ is \textsf{NP}-complete. Further works of \cite{Woeginger00} and~\cite{Brucker99} include the study of chain-like precedences. \cite{Woeginger00} provided a polynomial-time $2$-approximation algorithm for uniformly related parallel machines, i.e., for problem $Q|\operatorname{chains}|C_{\max}$, and~\cite{Brucker99} presented a linear-time algorithm for the problem $Q2|\operatorname{chains},p_j=1|C_{\max}$.

Subsequent research explored more elaborate forms of precedences. \cite{WIKUM199487} considered generalized precedence constraints where each constraint of the form $j_1<j_2$ is assigned two numbers $\ell_{j_1j_2}\leq u_{j_1j_2}$ and the time between beginning of $j_2$ and end of $j_1$ must be in the interval $[\ell_{j_1j_2},u_{j_1j_2}]$. They provided polynomial-time algorithms for some restricted cases of the precedence relation and (strong) \textsf{NP}-hardness results for the general cases.

\subsection*{Parameterized complexity}
Beyond classical polynomial-time solvability and \textsf{NP}-hardness, a more fine-grained understanding comes from the perspective of parameterized complexity. We refer the reader to \Cref{sec:preliminaries} for a comprehensive introduction to all terminology used within the parameterized complexity framework. Numerous works have studied the parameterized complexity of scheduling problems for parallel machines~\citep{Bodlaender95W2,ChenM18,Cieliebak04,vanBevernRolfSuchy2017,Bevern16,Bodlaender20,MallemHK22,ShabtayG24_ejor,BusingDM25}, single machine~\citep{FellowsM03,HermelinKSTW15,MallemHM24,DEWEERDT2021ejor,HERMELIN2019ejor,YU2025ejor}, in the high multiplicity encoding regime~\citep{KnopKLMO19}, in repetitive scheduling scenarios~\citep{HermelinMNPS25_ejor_repetitive}, or under processing time uncertainity~\citep{goldberg2026parametrized}. Below, we summarize those most relevant to our setting.

\cite{Cieliebak04} considered the problem of minimizing the number of identical parallel machines required to schedule jobs with release times and deadlines under the parameterizatin by \emph{slack} or \emph{flexibility} of the underlying jobs. The \emph{slack} of a job $j$ is the number $\delta_j=d_j-r_j-p_j$ and \emph{maximum slack} is $\delta = \max_j \delta_j$.
The \emph{flexibility} is $\lambda_j = \frac{{d}_j-r_j}{p_j}$. They show that the problem is polynomial-time solvable if the maximum slack is equal to $1$ and \textsf{NP}-complete even when the maximum slack is $2$ or the maximum flexibility $\lambda = \max_j\lambda_j$ is an arbitrary fixed number $\lambda > 1$. From the positive side, they provided an algorithm with running time $f(\delta,\mu)\cdot \operatorname{poly}(n)$ for the problem $P|r_j,d_j|*$, where $\mu$ is the maximum number of overlapping job windows (also known as \emph{path-width}), i.e., it is \FPT parameterized by the combined parameter $\delta$ and $\mu$.

Following the work of Cielibak et al., \cite{vanBevernRolfSuchy2017} showed that $P|r_j,d_j|*$ is weakly \textsf{NP}-hard for $m=2$ machines and any $\lambda > 1$ and strongly $\W{1}$-hard parameterized by the number of machines $m$. They present an \FPT algorithm for the combined parameter $m$ and $\delta$.

The work of \cite{Cieliebak04} was also extended to consider precedence constraints, and \cite{Hanen2024ParadependentTasks} proposed an \FPT algorithm parameterized by the combined parameter $\min(\delta,p_{\max})$ and $\mu$. Moreover, they established \pNP-hardness with respect to the single parameter $\mu$ on two machines. On the other hand, for a single machine, the problem $1|\operatorname{prec},r_j,d_j|C_{\max}$ was shown to be \FPT parameterized by $\mu$~\citep{MallemHM24}. \cite{TarhanHKCJ25} proposed a dynamic programming algorithm that lists all active schedules and provides \FPT algorithm parameterized by $p_{\max}$ and $\mu$ for various objective functions for parallel machines.

\cite{Bevern16} considered the parameterized complexity of scheduling with precedence constraints. They considered the parameter \emph{width} of the precedence relation, which is the maximum size of an independent set of jobs. They showed that the problem is weakly \textsf{NP}-hard for two machines and their reduction produces a precedence relation with $k=3$ chains. They further showed \W{2}-hardness of the problem $P2|\operatorname{prec}|C_{\max}$ parameterized by the width of the precedence relation, even with processing times equal to $1$ or $2$. This result was strengthened to {\sf XNLP}-hardness~\citep{BodlaenderMallem26}. On the positive side, they provided an \FPT algorithm parameterized by the width of the precedence relation combined with the parameter \emph{lag} -- the maximum time after which a job can start after its release time.

Regarding the parameter width, \cite{Mohring1989} also showed that the problem $P|\operatorname{prec},p_j=1 |C_{\max}$ is solvable in $O(\binom{w}{m}n^w)$ time where $w$ is the width, $m$ is the number of machines and $n$ the number of jobs, i.e., it is in $\XP$ parameterized by $w$ for constant number of machines. In general, the problem is {\sf XNLP}-complete~\cite{BodlaenderGNS24_complete_xnlp} parameterized by the combined parameter $m$ and $w$.

\cite{Bodlaender20} considered the problem \textsc{Chain Scheduling With Delays} where the input consists of $k$ chains of jobs and each chain has its release time and deadline. The precedences between jobs in each chain are generalized -- there is a required minimum (or exact) time (referred to as \emph{delays}) between two consecutive jobs in the resulting schedule. They showed that the problem is \W{1}-hard parameterized by $k$ and \W{$t$}-hard for every $t$ parameterized by \emph{thickness} -- maximum number of chains whose $[r_C,{d}_C)$ (i.e., release time---deadline interval) overlap.

Further parameterized complexity results include the works of \cite{MallemHK22,MallemHM24}. They considered the underlying interval graph induced by the time windows $[r_j,d_j)$ and parameterized the scheduling problem by path-width of the underlying graph and maximum delay. \cite{MallemHK22} showed that the problem $1|\operatorname{chains}(\ell_{i,j}),r_j,d_j,p_j=1|*$ is \pNP-hard parameterized by path-width with either minimum or exact delays, and that $P|\operatorname{chains}(\ell_{i,j}),r_j,d_j,p_j=1|*$ is \textsf{NP}-hard even when maximum delay is equal to $1$. On the other hand, the problem with general precedences becomes \FPT parameterized by the combined parameter maximum delay and path-width.

\cite{MallemHM24} then considered the parameter \emph{proper level}, which is the maximum number of time windows that can strictly include another time window on both ends. They show that $1|\operatorname{prec},r_j,d_j|C_{\max}$ is \FPT parameterized by proper level. Note that the proper level is upper bounded by path-width, which in turn is upper bounded by the maximum slack; hence, this also implies that the above problems are \FPT parameterized by path-width and maximum slack.

\section{Contribution and outline}
We examine the parameterized complexity of $1|\operatorname{prec},r_j,d_j|*$ and consider parameters closely related to RPA scheduling. We consider chain-like precedences and show that for any positive integers $p<q$, the problem is strongly {\sf XNLP}-hard parameterized by the number of chains, even when $p_j\in\{p,q\}$ and the instance has at most two distinct time windows, i.e., $\#(r_j-d_j)\leq 2$~(\Cref{thm:w2hardness}). In other words, we show that the problem is \pNP-hard parameterized by the combined parameter $\#p_j +p_{\max} + \#(d_j-r_j)$. Moreover, the problem does not admit an algorithm with running time $f(k)n^{o(k)}$ where $k$ is the number of chains. We then refine the construction and obtain \pNP-hardness for the combined parameter $\#p_j+\#(d_j-r_j)$ even when ensuring prec-consistent time windows (\Cref{cor:w2hardness_preccon}). As a direct corollary of a result of \cite{MallemHM24}, we show that the problem becomes polynomial-time solvable when the number of distinct time windows is $1$, i.e., $\#(d_j-r_j)=1$ (\Cref{thm:polytime_onewindow}). We observe that the problem also becomes \FPT when the release times, deadlines and processing time are chain-uniform by applying the results of \cite{KnopKLMO19} (\Cref{thm:knop_red}). Finally, complementing the hardness results, we show that $1|\operatorname{prec},r_j,d_j|*$ is in \XP parameterized by $w$ (i.e., the width of the precedence relation) (\Cref{thm:xp_width}). In a slightly restricted scenario, where the maximum processing time can be bounded by $n^{f(w)}$ for some computable function $f$, we show that the problem $1|\operatorname{prec},r_j,d_j|*$ is complete for {\sf XNLP}.

We note that the \XP algorithm runs in $n^{O(w)}$ time and uses the same amount of space. An algorithm with running time $n^{o(w)}$ is ruled out under the assumption of the Exponential Time Hypothesis. Under the slice-wise polynomial space conjecture by \cite{PilipczukW18_xnlp_conjecture} (see \Cref{conjecture:xnlp}) the space complexity cannot be improved to $g(w)\cdot \operatorname{poly}(n)$ while keeping the running time at $n^{g(w)}$ for some computable function $g$. We summarize our contributions in the context of known results in Table~\ref{tab:table_results}.

\begin{table}
    \centering
    \caption{Summary of the key results and related problems. The first part of the table summarizes known results for parallel machines scheduling problems. The second part summarizes known results for single machine scheduling problems. The third part highlights the main contributions of this paper and the last part indicates open problems. The $+$ in the parameters column indicates that the result holds for any combination of these parameters. For the formal definitions of the parameters refer to \Cref{sec:preliminaries} and for the discussion of open problems refer to \Cref{sec:conclusion}.
    }
    \resizebox{1\textwidth}{!}{%
    \begin{tabular}{|c|c|c|c|c|}
    \toprule
    Problem & Parameter(s) & Result & Reference \\
    \midrule
    $P|\operatorname{prec},p_j=1|C_{\max}\leq 3$ & - & \textsf{NP}-hard & \cite{Lenstra78} \\
    
     $P2|\operatorname{chains}(k)|C_{\max}$ & $k$ & \pNP-hard & \cite{Bevern16} \\
    
    $P2|\operatorname{chains}(k),r_j,d_j|C_{\max}$ & $\mu + k$ & \pNP-hard & \cite{Hanen2024ParadependentTasks} \\

    $P|\operatorname{prec},p_j=1|C_{\max}$ & width $w$ of $\operatorname{prec}$ & \textsf{XP} for $O(1)$ machines &  \cite{Mohring1989} \\

    $P|\operatorname{chains},r_j,d_j,p_j=1|C_{\max}$ & - & \textsf{P} &\cite{Dror97} \\

    \midrule
    $1|r_j,d_j|*$ & - & \textsf{NP}-hard &  \cite{GareyJ1979} \\
    $1|\operatorname{prec},r_j|C_{\max}$ & - & \textsf{P}& \cite{Lawler1973} \\
    $1|\operatorname{prec},d_j|C_{\max}$ & - & \textsf{P} &\cite{Lawler1973}\\
    
    $1|\operatorname{prec},r_j,d_j,p_j=p|C_{\max}$ & - & $\textsf{P}$ &\cite{Simons1978AFA} \\

    $1|r_j,d_j,p_j\in \{1,p\}|C_{\max}$ & - & $\textsf{P}$  &\cite{Sgall12} \\
    $1|r_j,d_j|C_{\max}$ & $\#(r_j,d_j,p_j)$ & \FPT & \cite{KnopKLMO19}\\

    $1|\operatorname{prec},r_j,d_j|C_{\max}$ &  $\delta$  &\FPT &\cite{MallemHM24}\\
    $1|\operatorname{prec},r_j,d_j|C_{\max}$ &  $\mu$ & \FPT &\cite{ MallemHanenMunier:22:ParameterizedSingle}\\

    $1|\operatorname{prec},r_j,d_j,p_j\in \{1,p\}|*$ & $p_{\max}$ & \W{1}-hard  &\cite{LeslieSh92} \\
    
    $1|\operatorname{prec},r_j,d_j,p_j\in \{1,p\}|*$ & $\#r_j+\#d_j+\#p_j$ & \pNP-hard  &\cite{LeslieSh92} \\

    $1|r_j,d_j,p_j\in\{p,q\}|*$ & - & \textsf{NP}-hard & \cite{Elffers17AuxPQ} \\

\midrule

    $1|\operatorname{chains}(k),r_j,d_j,p_j\in\{p,q\}|*$ & $k$ & {\sf XNLP}-hard & \Cref{thm:w2hardness} \\
    $1|\operatorname{chains},r_j,d_j,p_j\in\{p,q\}|*$ & $\#p_j+p_{\max} + \#(d_j-r_j)$ & \pNP-hard & \Cref{thm:w2hardness} \\
    
    $1|\operatorname{chains},r_j,d_j,\operatorname{prec}\text{-con}|*$  & $\#p_j+\#(d_j-r_j)$ & \pNP-hard & \Cref{cor:w2hardness_preccon} \\
    $1|\operatorname{prec},r_j,d_j|C_{\max}$ & width $w$ of $\operatorname{prec}$ & \XP  &\Cref{thm:xp_width} \\

    $1|\operatorname{prec},r_j,d_j|C_{\max}$ & width $w$ of $\operatorname{prec}$ & \XNLP when $p_{\max}\leq n^{f(w)}$ &\Cref{thm:xnlp_completeness_width} \\
    
    $1|\operatorname{prec},r_j,d_j,\operatorname{prec}\text{-con}|C_{\max}$ & $\#(d_j-r_j)=1$ & \textsf{P} & \Cref{thm:polytime_onewindow} \\
   $1|\operatorname{chains}(k),r_Q,{d}_Q,p_j=p_Q|*$ & $k$ & \FPT &\Cref{thm:knop_red}\\

    \midrule
    $1|\operatorname{chains},r_j,d_j|C_{\max}$ & $\#(r_j,d_j,p_j)$ & ? & -\\

    $1|\operatorname{chains}(k),r_j,d_j,p_j=p_Q|*$ & $k$ & ? &  -\\
    
    \bottomrule
    \end{tabular}%
    }

    \label{tab:table_results}
\end{table}
The paper is organized as follows. In \Cref{sec:preliminaries} we review the basic notions from parameterized complexity and scheduling. In \Cref{sec:hardness}, we present our hardness results, identifying parameters under which the problem remains intractable, while in \Cref{sec:algorithms}, we provide algorithmic results, including \FPT, \XP and \XNLP algorithms. Based on the complexity results, \Cref{sec:rpa_pc_relation_discussion} describes strategies for designing scheduling algorithms for RPA systems. We conclude with open questions and future research directions in \Cref{sec:conclusion}.

\section{Preliminaries}\label{sec:preliminaries}
\subsection*{General notation}
For nonnegative integers $a,b$ we let $[a,b]:=\{a,a+1,\ldots,b-1,b\}$ and $[a]:=[1,a]$. In particular, $[0]=\emptyset$. A \emph{strict partial order} is a pair $(J,\prec)$ where 
$J$ is a finite set and $\prec$ is irreflexive and transitive binary relation on~$J$. Let $(J,\prec)$ be a strict partial order, a set $J'\subseteq J$ is \emph{$\prec$-independent} if $j_1\not\prec j_2,j_2\not\prec j_1$ for any two distinct $j_1,j_2\in J'$. The \emph{width} of $(J,\prec)$, denoted $w$, of $\prec$ is the size of the largest $\prec$-independent set in $J$.

\subsection*{Parameterized complexity}
Parameterized complexity is a framework aiming to study the computational complexity of problems according to their inherent difficulty with respect to one or more \emph{parameters} of the input or output. This allows for a finer analysis of \textsf{NP}-hard problems than in classical complexity. Let~$\Sigma$ be a fixed alphabet and $\Sigma^*$ the set of all strings over $\Sigma$. A \emph{parameterized problem} is $L\subseteq \Sigma^*\times \mathbb{N}$. An instance of a parameterized problem is a pair $(x,k)$ where $x$ denotes the input instance and $k$ is the \emph{parameter value}. We denote by $|(x,k)|$ the bitsize of the instance $(x,k)$. A parameterized problem is said to be \emph{fixed-parameter tractable} or (in the class) \textsf{FPT} if it is solvable in $f(k)\cdot \operatorname{poly}(|(x,k)|)$ time for some computable function $f$. A parameterized problem is (in the class) \textsf{XP} if it can be solved in $|(x,k)|^{f(k)}$ time for some computable function $f$. Showing that a problem $L$ is in \XP parameterized by parameter $k$ is equivalent to claiming that for each fixed $k$ there is a polynomial-time algorithm for $L$, where the degree of the polynomial might possibly depend on $k$. Showing that a problem $L$ is in \FPT means that for each fixed $k$ there is a polynomial-time algorithm where only the multiplicative constant in the running time may depend on $k$ (i.e., the degree of the polynomial is a constant independent of $k$). If a parameterized problem $L$ is shown to be \textsf{NP}-hard even if the parameter $k$ in the input instance is bounded by a constant, then the problem is said to be \textsf{para-NP}-hard. In such a case, the parameterized problem  $L$ cannot belong even to the class \XP unless $\textsf{P}=\textsf{NP}$. Central to the parameterized complexity framework is the following hierarchy of classes:
\[
    \textsf{FPT}\subseteq \textsf{W[1]}\subseteq \textsf{W[2]}\subseteq\cdots \subseteq \textsf{W[SAT]}\subseteq \textsf{W[P]}\subseteq\textsf{XNLP}\subseteq \textsf{XP}
\]
where all inclusions are conjectured to be strict. The only known strict inclusion is \mbox{$\textsf{FPT}\subsetneq \textsf{XP}$}. 
We formally define the class \textsf{XNLP} below. We omit the rigorous definitions for the \textsf{W}-classes -- the technical definitions are less crucial for our purposes than the structural role these classes play. For rigorous definitions, we refer the reader to standard parameterized complexity textbooks (e.g.~\cite{CyganParaAlg2015,DowneyF13}). Much like the role {\sf NP} plays in classical complexity, these classes allow us to classify problems that seem to be beyond the reach of \FPT. If a problem is shown to be hard for some class above \FPT in this hierarchy (e.g. \W{1}-hard or {\sf XNLP}-hard), this is a strong evidence that the problem is likely not in \FPT, i.e., we cannot expect to find an algorithm with running time $f(k)\cdot \operatorname{poly}(|(x,k)|)$ assuming some standard complexity-theoretical conjectures such as $\FPT \neq \W{1}$~(see \cite{DowneyF99_para_complexity}).

As in classical complexity, where the notion of polynomial-time reduction is used to show \textsf{NP}-hardness, the parameterized complexity framework relies on \emph{parameterized reductions} to establish evidence of difficulty.

Let $L,L'$ be parameterized problems. A \emph{parameterized reduction} from $L$ to $L'$ is an algorithm~$\mathcal{A}$ that, given an instance $(x,k)$ of problem $L$, computes an instance $(x',k')$ of $L'$ such that $(x,k)\in L\Leftrightarrow (x',k')\in L'$, the running time of $\mathcal{A}$ is $f(k)\cdot \operatorname{poly}(|(x,k)|)$ and $k'\leq g(k)$ for some computable functions $f,g$. Parameterized problem $L'$ is said to be \textsf{W[$t$]}-hard if for all problems $L\in \textsf{W[$t$]}$ there is a parameterized reduction from $L$ to $L'$. Showing that a parameterized problem $L'$ is \textsf{W[$t$]}-hard for some $t$ shows that $L'$ is unlikely to be \textsf{FPT}.

\subsection*{Logspace}
The classes {\sf L} (deterministic logspace) and {\sf NL} (nondeterministic logspace) are classes of (unparameterized) decision problems solvable by a deterministic, resp. nondeterministic Turing machine with $O(\log |x|)$ auxiliary memory, where $|x|$ denotes the number of bits to encode the input instance~$x$. The class {\sf coNL} consists of problems whose complement is in {\sf NL}, i.e., ${\sf coNL}=\{L\mid \overline{L}\in {\sf NL}\}$.

\begin{definition}[{Logspace Reduction~\cite[Definition 4.14]{AroraBarak09_complexity_modern_approach}}]
    Let $f\colon \{0,1\}^*\to \{0,1\}^*$ be a polynomially bounded function. We say that $f$ is \emph{implicitly logspace computable} if the languages $L_f=\{(x,i)| f(x)[i]=1\}$ and $L'_f=\{(x,i)\mid i\leq |f(x)|\}$ are in {\sf L}. A \emph{logspace reduction} from problem $A$ to problem $B$ is an implicitly logspace computable function $f\colon \Sigma^*\to \Sigma^*$ such that $x\in A$ if and only if $f(x)\in B$.
\end{definition}
Essentially the conditions on the languages $L_f$ and $L'_f$ state that for any index $i$ we can in logspace compute the $i$-th bit of the output instance given that $i\leq |f(x)|$, i.e., $i$ is in the range of valid indices of the output ($L_f\in \LOG)$. The fact that $i$ is in the range of valid indices of the output must also be verifiable in logspace ($L'_f\in \LOG)$.

In the decision problem \textsc{STcon}, the input is a directed graph $G$ and two vertices $s,t\in V(G)$ and the task is to decide whether $t$ is reachable from $s$ in $G$. It is known that $\textsc{Stcon}\in {\sf NL}$~\citep{AroraBarak09_complexity_modern_approach}. We utilize a generalization of \textsc{STcon} which we refer to as \textsc{STcon+}. In the problem \textsc{STcon+}, we are given a directed graph $G$, starting vertex $s$ and a set $M\subseteq V(G)$. The task is to decide whether there exists a path from $s$ to a vertex of $M$.

\begin{observation}\label{obs:stcon_plus}
    \textsc{STcon+} is in {\sf NL}
\end{observation}
\begin{proof}
    We can nondeterministically guess the vertex $m\in M$ and then we solve the standard \textsc{STcon} problem which is in {\sf NL}.
\end{proof}
We use the following corollary of a result of~\cite{Immerman88} and~\cite{Szelepcsenyi88}.
\begin{corollary}[\cite{Immerman88,Szelepcsenyi88}]\label{cor:immerman_szelepczenyi}
    ${\sf NL}={\sf coNL}$
\end{corollary}

\subsection*{Circuits}

\begin{definition}[Boolean circuit]
    Let $n,m\in\mathbb{Z}^+$. An \emph{$n$-input $m$-output boolean circuit} $C$ is a directed acyclic graph $G$ with exactly $n$ sources (vertices of in-degree $0$) $i_1,i_2,\ldots,i_n$ (\emph{inputs}) and exactly $m$ sinks (vertices of out-degree $0$) $o_1,o_2,\ldots,o_m$ (\emph{outputs}). Any non-source vertices are \emph{gates} and are labelled with either {\tt AND}, {\tt OR} or {\tt NOT}. The {\tt NOT} gate has in-degree $1$.

    Given $x=x_1x_2\cdots x_n\in \{0,1\}^n$, we define the valuation function $\operatorname{val}_C\colon V(G)\to \{0,1\}$ as follows
    \begin{enumerate}
        \item $\operatorname{val}_C(i_j)=x_j$
        \item $\operatorname{val}_C(v)=\displaystyle\bigwedge_{i=1}^k\operatorname{val}_C(u_i)$ if $v$ is labelled with {\tt AND} and $u_1,u_2,\ldots,u_k$ are all the in-neighbors of~$v$.
        \item $\operatorname{val}_C(v)=\displaystyle\bigvee_{i=1}^k\operatorname{val}_C(u_i)$ if $v$ is labelled with {\tt OR} and $u_1,u_2,\ldots,u_k$ are all the in-neighbors of~$v$.
        \item $\operatorname{val}_C(v)=1-\operatorname{val}_C(u)$ if $v$ is labelled with {\tt NOT} and $u$ is the in-neighbor of $v$.
    \end{enumerate}
    The string $y_1y_2\ldots y_m\in\{0,1\}^m$ where $y_i=\operatorname{val}_C(o_i)$ is the \emph{output} of the computation of $C$ on the input $x$ and is denoted $C(x)$.

    The \emph{depth} of $C$ is the longest path starting from an input vertex and ending in an output vertex. The \emph{size} of $C$ is the number of vertices of $G$. The \emph{fan-in} of a gate $v\in V(G)$ is the number of in-neighbors of $v$. The \emph{fan-in} of $C$ is the maximum fan-in of a gate of $C$.
\end{definition}

\begin{definition}
    A (partial) function $f\colon \{0,1\}^*\to \{0,1\}^*$ is in the class ${\sf FNC}^1$ if there exists a family of boolean circuits $(C_n)_{n\in\mathbb{N}}$ such that
    \begin{enumerate}
        \item the fan-in of each $C_n$ is $2$,
        \item the depth of each $C_n$ is $O(\log n)$,
        \item the size of $C_n$ is $\operatorname{poly}(n)$, and
        \item for every $x\in \operatorname{dom}(f)$ we have $C_{|x|}(x)=f(x)$.
    \end{enumerate}

    The class \emph{logspace uniform ${\sf FNC}^1$} consists of such functions $f$ as above, where in addition, we can in logspace compute the following functions:
    
    \begin{enumerate}
        \item  ${\tt SIZE}(n)$, which returns the size of $C_n$ (in binary),
        \item ${\tt TYPE}(n,i)$ which given $n$ and $i$, outputs the type of the $i$-th gate in $C_n$, which is one of $\{\wedge,\vee,\neg,\operatorname{INPUT}(1),\operatorname{INPUT}(2),\ldots,\operatorname{INPUT}(m)\}$, where $m$ is the number of output gates of $C_n$ and
        \item ${\tt EDGE}(n,i,j)$ which returns $1$ if there is a directed edge from $i$-th to $j$-th node in $C_n$ and $0$ otherwise.
    \end{enumerate}
    
\end{definition}

\begin{lemma}[{\cite[Theorem 1.4.]{ChiuDL01_logspace_addition}}]\label{lem:addition_of_n_numbers}
    Addition of $n$ numbers, each with $n$ bits can be computed by a logspace uniform family of ${\sf FNC}^1$ circuits.
\end{lemma}

\begin{lemma}\label{lem:ftc0_is_logspace_red}
    A reduction $f\colon \{0,1\}^*\to \{0,1\}^*$ that is in logspace uniform ${\sf FNC}^1$ is a logspace reduction.
\end{lemma}
\begin{proof}
    We must show that we can compute each bit of the output independently in logspace. We perform depth-first-search from the corresponding output gate in the opposite direction of the edges of the circuit. Recall that the depth is $O(\log n)$. To store the path from the root to the current gate, we only store the index of the current gate ($O(\log n)$ bits) and only one bit for every successor gate on the path. For the successor gates we only store the index of the neighbor into which we are currently recursing. As the circuit has fan-in $2$, this needs constantly many bits per gate. Note that this information is sufficient for reconstructing the next gate to recurse into during the search using the {\tt EDGE} function in logspace.
\end{proof}

\subsection*{The class {\sf XNLP}}
The class  {\sf XNLP} is the class of parameterized problems solvable in $f(k)\operatorname{poly}(|x,k|)$ time and $f(k)\log |(x,k)|$ space by a nondeterministic Turing machine. Note that the Turing machine has read-only input tape and write-only output tape and read/write working tape. The space complexity is then the maximum memory used on the working tape throughout the computation.

When dealing with hardness or completeness for the class {\sf XNLP}, we need a stricter variant of parameterized reduction. A \emph{parameterized logspace reduction} (or \emph{pl-reduction}) is the same as logspace reduction but the memory usage is limited to $O(h(k)+\log |(x,k)|)$ for some computable function $h$. When referring to hardness or completeness with respect to {\sf XNLP}, it is always meant under pl-reductions. Pl-reductions have the following expected behaviour for {\sf XNLP} problems:
\begin{lemma}\label{lem:xnlp_pl_reductions}
    Let $L,L'$ be parameterized problems. If $L'\in{\sf XNLP}$ and $L$ is pl-reducible to $L'$, then $L\in{\sf XNLP}$.    
\end{lemma}
The proof of \Cref{lem:xnlp_pl_reductions} can be carried out similarly as the proof of Lemma 4.15 in \cite{AroraBarak09_complexity_modern_approach}.
We note that if a problem is {\sf XNLP}-hard, then it is also \W{t}-hard for every $t\geq 1$, and {\sf XNLP} is contained in the class {\sf XP}. See also~\cite{BodlaenderGNS24_complete_xnlp} for basic results about {\sf XNLP}. Pilipczuk and Wrochna stated the following conjecture regarding {\sf XNLP}-hard problems. See also~\cite{BodlaenderGJJL25_swxp_conjecture2}.
\begin{conjecture}[{\sf Slice-wise Polynomial Space Conjecture}~\citep{PilipczukW18_xnlp_conjecture}]\label{conjecture:xnlp}
{\sf XNLP}-hard problems do not admit algorithms with running time $n^{f(k)}$ and space complexity $f(k)\operatorname{poly}(n)$, where $k$ is the parameter, $n$ is the input size and $f$ is a computable function.
\end{conjecture}
Another conjecture related to {\sf XNLP} is that of Bodlaender stating that {\sf XNLP}-hard problems probably do not belong to the class \textsf{W[P]}~(see \cite{BodlaenderGNS24_complete_xnlp}).

\subsection*{Strong hardness}
For a class of problems $\mathcal{C}$ (parameterized or unparameterized), we say that a problem $L$ is \emph{strongly} $\mathcal{C}$-hard (under suitable reductions) if $L$ is $\mathcal{C}$-hard (under suitable reductions) even when all numbers in the instances of $L$ are encoded in unary.

\subsection*{Exponential Time Hypothesis}
The \emph{Exponential Time Hypothesis} (ETH) asserts, roughly speaking, that there is no algorithm for \textsc{3Sat} in time $2^{o(n')}$ where $n'$ is the number of variables in the input formula \citep{ImpagliazzoPZ01journal}. As a corollary, we use the result of \cite{ChenHKX06journal} that there is no $f(k)n^{o(k)}$ algorithm for \textsc{Dominating Set} parameterized by the solution size for any computable function~$f$ unless ETH fails.

\subsection*{Scheduling}

We adopt the standard three-field notation $\alpha|\beta|\gamma$, introduced by~\cite{GRAHAM1979287}. In the three-field notation, $\alpha$ represents the scheduling environment, which in this paper is the number of machines: one machine ($1$), or $m$ unrelated parallel machines ($Pm$), or unspecified number of unrelated parallel machines ($P$). The $\beta$ field represents the job characteristics. These can be for example: precedences ($\operatorname{prec}$), chains ($\operatorname{chains}$, release times of jobs ($r_j$), deadlines of jobs ($d_j$), specifications of processing times ($p_j=1$,$\ldots$). Finally, the $\gamma$ field represents the objective of the problem. In this paper, we consider either minimizing the makespan $C_{\max}$, which is the maximum completion time of a job, or an $*$ indicating that we are only interested in feasible schedule.

In this paper we consider non-preemptive scheduling on a single machine. On this machine, our aim is to process all jobs. The set of jobs is denoted $J$. Each job $j\in J$ has a release time $r_j\in \mathbb{Z}^+_0$, processing time $p_j\in\mathbb{Z}^+$ and deadline $d_j\in \mathbb{Z}^+_0$. The \emph{slack} of a job $j\in J$ is the number $\delta_j={d}_j-r_j-p_j$. The \emph{maximum slack} of $J$ is $\delta:=\max_{j\in J}\delta_j$. The \emph{window size} of job $j$ is the number $d_j-r_j$ and $\#(d_j-r_j)$ is the \emph{number of window sizes} defined as $\#(d_j-r_j)=|\{d_j-r_j\mid j \in J\}|$.
A \emph{type} of job $j$ is the triple $(r_j,d_j,p_j)$. The \emph{number of job types}, denoted $\#(r_j,d_j,p_j)$ is defined as $\#(r_j,d_j,p_j)=|\{(r_j,d_j,p_j)\mid j \in J\}|$.

A \emph{precedence constraint} is a strict partial order $\prec$ on $J$. If $\prec$ is a disjoint union of chains, i.e., $J=Q_1\cup Q_2\cup \cdots \cup Q_k$ and for every $j,j'\in Q_i$ we have $j\prec j'$ or $j' \prec j$ and for every $j\in Q_{i},j'\in Q_{i'}$ for $i\neq i'$ we have $j\nprec j'$ and $j'\nprec j$. We write $p_Q,r_Q,d_Q$ to represent the processing times, release times, or deadlines that are only dependent on the chain $Q\in \{Q_1,Q_2,\ldots, Q_k\}$. In such a case, we say that the processing times (or release times, deadlines) are \emph{chain-uniform}. We say that the instance is \emph{$\operatorname{prec}$-consistent} if whenever $j\prec j'$, then $r_{j'}\geq r_j+p_j$ and $d_j\leq d_{j'}-p_{j'}$.

Equivalently, we think of a partial order $\prec$ on $J$ as a directed acyclic graph whose vertex set corresponds to $J$ and the edges correspond to the edges of the Hasse diagram of $\prec$.

A \emph{schedule} for $J$ is a function $\sigma \colon J \to \mathbb{Z}^+_0$ assigning each job $j\in J$ a start time $\sigma_j:=\sigma(j)$. 
The \emph{completion time} of a job $j\in J$ under schedule~$\sigma$ is the time point $C_j=\sigma_j+p_j$. A schedule~$\sigma$ is \emph{feasible} if:
\begin{itemize}
    \item no job starts before its release: $\forall j \in J:\sigma_j\geq r_j$,
    \item each job is processed by its deadline: $\forall j \in J: \sigma_j+p_j\leq {d}_j$,
    \item no two jobs are processed at the same time: for any $j,j'\in J,j\neq j'$ we have $\sigma_{j}\geq \sigma_{j'}+p_{j'}$ or $\sigma_{j'}\geq \sigma_{j}+p_{j}$,
    \item the schedule respects the precedence constraints: For any $j,j'\in J$ if $j\prec j'$, then $\sigma_j\leq \sigma_{j'}$.
\end{itemize}

The \emph{makespan} of $\sigma$ is $C_{\max}:=\max_{j\in J}C_j$. The makespan of an empty schedule is $0$.

\section{Hardness results}\label{sec:hardness}
In this section we show strong {\sf XNLP}-hardness of $1|\operatorname{chains}(k),r_j,d_j|*$ parameterized by $k$ -- the number of chains. As a by-product of our reduction, we also obtain \pNP-hardness for the combined parameter $\#p_j+p_{\max}+\#(d_j-r_j)$. We first present the problem that we will reduce from.

For a word $w\in\Sigma^*$ let $w[i]$ denote the character at position $i$ in $w$ (indexed from $1$) and $|w|$ the length of $w$. Let $u_1,u_2\in \Sigma^*$ be two words. The \emph{shuffle product} of $u_1$ and $u_2$, denoted $u_1\shuffle u_2$, is the set of all $\frac{(|u_1|+|u_2|)!}{|u_1|!|u_2|!}$ ways of interleaving the letters of $u_1$ and $u_2$ while respecting their order in both strings. See \Cref{ex:shuffle_product}:
\begin{example}\label{ex:shuffle_product}
    ${\color{red}ab} \shuffle {\color{blue}cd}=\{{\color{red}ab}{\color{blue}cd},{\color{red}a}{\color{blue}c}{\color{red}b}{\color{blue}d},{\color{red}a}{\color{blue}cd}{\color{red}b},{\color{blue}cd}{\color{red}ab},{\color{blue}c}{\color{red}ab}{\color{blue}d},{\color{blue}c}{\color{red}a}{\color{blue}d}{\color{red}b}\}$.
\end{example}
Formally, we have $v\in u_1\shuffle u_2$ if and only if there are strictly increasing functions $f_1\colon [|u_1|]\to [|v|]$ and $f_2\colon [|u_2|]\to [|v|]$ mapping positions of $u_1$ and $u_2$ to positions of $v$ such that for every index $x$ we have $v[f_1(x)]=u_1[x]$ and $v[f_2(x)]=u_2[x]$, for any $x\in[|u_1|],y\in[|u_2|]$ we have $f_1(x)\neq f_2(y)$ and for every $i\in[v]$ we have either $i\in f_1([|u_1|])$ or $i\in f_2([|u_2|])$.
More generally, for $\ell$ strings $u_1,u_2,\ldots,u_{\ell}\in \Sigma^*$ we have $v\in u_1\shuffle u_2 \shuffle \cdots \shuffle u_\ell$ if there are increasing functions $f_i\colon [|u_i|]\to [|v|]$ that map the positions of $u_i$ to the positions of $v$ such that $v[f_i(x)]=u_i[x]$ and the images of $f_i$ form a partition of $[|v|]$.

The decision problem we will be using in our reduction is \textsc{Binary Shuffle Product}, formally defined as follows:

\begin{center}
\begin{tabular}{|r|l|}
    \hline
    &\textsc{Binary Shuffle Product}
     \\\hline {Input}: & Words $u_1,u_2,\ldots,u_\ell,v$ over a binary alphabet $\Sigma$. \\\hline
     {Output}:& Is $v\in u_1\shuffle u_2 \shuffle \cdots \shuffle u_\ell$? \\\hline
\end{tabular}
\end{center}

It is known that for unbounded $\ell$ \textsc{Binary Shuffle Product} is \textsf{NP}-hard~\citep{WarmuthH84} and \W{2}-hard parameterized by $\ell$~\citep{Bevern16}. Recently, \cite{BodlaenderMallem26} showed that it is {\sf XNLP}-complete. We begin by describing a (parameterized) reduction from \textsc{Binary Shuffle Product} to $1|\operatorname{chains}(k),r_j,d_j,p_j\in \{p,q\}|*$ which we later use with modifications.

\begin{construction}\label{construction:w2hardness}
    Let $p,q$ be any two distinct positive integers, without loss of generality, let $p<q$. Let  $(u_1,u_2,\ldots,u_\ell,v)$ be an instance of \textsc{Binary Shuffle Product} over the binary alphabet $\Sigma=\{p,q\}$. We treat characters $p,q$ both as symbols of $\Sigma$ and as the integers $p,q$. For string $x\in \Sigma^*$ and $a\in \Sigma$ let $|x|_a$ denote the number of occurences of $a$ in $x$. If $|v|_a\neq \sum_{i=1}^{\ell}|u_i|_a$ for some $a$, we output a trivial no-instance of $1|\operatorname{chains}(k),r_j,d_j,p_j\in \{p,q\}|*$. Otherwise, we create the resulting instance as follows.

    The number of chains is $k=\ell + 1$. The first chain contains $|v|$ many \emph{guard} jobs $g_0\prec g_1\prec \cdots \prec g_{|v|}$ (see \Cref{fig:para_reduction} for an illustration). All guard jobs have processing time equal to $p$ and $d_{g_i}-r_{g_i}=p$. The release times are given by $r_{g_i}=p\cdot i + \sum_{t=1}^iv[t]$ for all $i\geq 0$. Observe that this creates a gap between the release of $g_i$ and the deadline of $g_{i-1}$ of size exactly $v[t]$ for every $t\in[|v|]$.

    Next, for every string $u_i$ we create a chain of $|u_i|$ jobs $x_1^i\prec x_2^i \prec \cdots \prec x_{|u_i|}^i$ corresponding to the letters of $u_i$. The release time of all these jobs is $0$ and the deadline is $(|v|+1)\cdot p+\sum_{t=1}^{|v|}v[t]$ (which is the same as $d_{g_{|v|}}$). The processing times are $p_{x_j^i}=u_i[j]$.
\end{construction}

\begin{figure}
        \centering
        \begin{tikzpicture}
        \draw[thick, ->] (-0.5,0) -- (10,0) node[below] {$t$};

        \foreach \x in {0,1,2,3,4,5,6,7,8,9} {
            \draw[thick] (\x,-0.3) -- (\x,0.3);
            \node[below] at (\x,-0.3) {\x};
        }
        \foreach \a/\b/\label in {0/1/$g_0$, 2/3/$g_1$, 5/6/$g_2$, 8/9/$g_3$} {
            \fill[red,opacity=0.3] (\a,0.4) rectangle (\b,0);
            \node[above] at ({(\a+\b)/2},0.5) {\label};
        }
        \end{tikzpicture}    
        \caption{The structure of the guard jobs in the reduction from \Cref{construction:w2hardness} for $\Sigma=\{p,q\}$ where $p=1,q=2$ and word $v=122$. Notice that the gap between the release of $g_i$ and the deadline of $g_{i-1}$ is exactly $v[i]$. The red rectangles represent the time windows $[r_{g_i},d_{g_i})$ of the guard jobs and it holds that $p=d_{g_i}-r_{g_i}=1.$}
        \label{fig:para_reduction}
\end{figure}

\begin{lemma}\label{lem:construction_correctness}
    Let $(u_1,u_2,\ldots, u_\ell,v)$ be an instance of \textsc{Binary Shuffle Product} and $\mathcal{I}$ the instance of $1|\operatorname{chains}(k),r_j,d_j|*$ produced by \Cref{construction:w2hardness}. Then $\mathcal{I}$ is yes-instance if and only if $(u_1, u_2, \ldots, u_{\ell}, v)$ is yes-instance.
\end{lemma}
\begin{proof}
    Note that since $d_{g_i}-r_{g_i}=p$, in any feasible schedule $\sigma$ we get $\sigma_{g_i}=r_{g_i}$. We shall refer to the time slot between $g_{i-1}$ and $g_i$ as the $i$-th time slot.

   On the one hand, if $v\in u_1\shuffle u_2 \shuffle \cdots \shuffle u_k$, then we schedule the guard jobs in the only possible way ($\sigma_{g_i}=r_{g_i}$). The jobs $x_j^i$ are scheduled in the $f_{i}(j)$-th time slot, i.e., between the jobs $g_{f_{i}(j)-1},g_{f_{i}(j)}$.

    On the other hand, a feasible schedule induces for each word $u_i$ the function $f_{i}\colon [|u_i|]\to [|v|]$ mapping the indices $j$ of $u_i$ to the indices of $v$ as follows. The value $f_{i}(j)$ is the number of non-guard jobs scheduled before the job $x_j^i$ plus one (i.e., including the job $x_j^i$ itself). It is not hard to verify that such $f_{i}$'s are valid functions to witness $v\in u_1\shuffle u_2 \shuffle \cdots \shuffle u_{\ell}.$ 
\end{proof}

Notice that the instance $\mathcal{I}$ from \Cref{construction:w2hardness} can be clearly produced in polynomial time even if all the numbers are encoded in unary. Moreover, we have $\#(d_j-r_j)\leq 2$ and the processing times satisfy $p_j\in \{p,q\}$. We obtain:

\begin{theorem}\label{thm:w2hardness}
    For any two distinct positive integers $p,q$, the problem $1|\operatorname{chains}(k),r_j,d_j,p_j\in \{p,q\}|*$ is strongly {\sf XNLP}-hard parameterized by $k$ and \pNP-hard parameterized by any combination of $\#p$, $p_{\max}$, and $\#(d_j-r_j)$. Moreover, the problem cannot be solved in $f(k)n^{o(k)}$ time for any computable function $f$ unless ETH fails.
\end{theorem}
\begin{proof}
    \textsc{Binary Shuffle Product} is {\sf XNLP}-hard parameterized by $\ell$ by the result of~\cite{BodlaenderMallem26}. Assume that $p<q$ and use \Cref{construction:w2hardness}. Note that all the numbers can indeed be encoded in unary and the running time remains even polynomial in the bitsize of the input instance of {\sc Binary Shuffle Product}. The correctness of the reduction follows from \Cref{lem:construction_correctness}. It remains to argue that the construction is a pl-reduction. Note that we only need an auxiliary memory to compute the release times $r_{g_i}$. Note that $p,q$ are treated as constants hence we need $O(1)$ bits to represent them and the largest number in the resulting instance is of size $(|v|+1)\cdot p + \sum_{t=1}^{|v|}v[t]$ which is at most $O(|v|q)$ and thus needs thus $O(\log |v| + \log q)=O(\log |v|)$ auxiliary bits on the work tape to be computed. Since the instance $(u_1,u_2,\ldots,u_{\ell},v)$ of \textsc{Binary Shuffle Product} can be encoded using $O(|v|)$ bits, this shows that the reduction is a pl-reduction.

    Observe that in the resulting instance, we have $\#p=2,p_{\max}=\max\{p,q\}=O(1)$, and $\#(d_j-r_j)\leq 2$. Note that we can attain $p_{\max}=2$ if we chose $\{p,q\}=\{1,2\}$. This shows \pNP-hardness for the combination of $\#p$, $p_{\max}$, and $\#(d_j-r_j)$.
    
    It remains to show the $f(k)n^{o(k)}$ complexity lower bound. We note that the reduction showing \W{2}-hardness of \textsc{Binary Shuffle Product} parameterized by $\ell$ of~\cite{Bevern16} is from the \textsc{Dominating Set} problem parameterized by the solution size $k'$ and the resulting parameter of the \textsc{Binary Shuffle Product} instance $\ell$ is linear in the parameter of the \textsc{Dominating Set} (in fact $\ell = k'+3$). The parameter of the instance of $1|\operatorname{chains}(k),r_j,d_j,p_j\in \{p,q\}|*$ from \Cref{construction:w2hardness} is also linear in the original parameter of \textsc{Binary Shuffle Product} ($k=\ell + 1$). Combining this with the result of \cite{ChenHKX06journal} about the $f(k')n^{o(k')}$ lower bound for \textsc{Dominating Set} parameterized by solution size $k'$, we obtain the same lower bound for $1|\operatorname{chains}(k),r_j,d_j|*$ parameterized by $k$.
\end{proof}

\begin{corollary}
    \begin{enumerate}[label=\alph*)]
        \item Under the slice-wise polynomial space conjecture~(\Cref{conjecture:xnlp}), there is no algorithm for $1|\operatorname{chains}(k),r_j,d_j|*$ with running time $n^{f(k)}$ that uses $f(k)\operatorname{poly}(n)$ space for any computable function $f$.
        \item Under the Exponential Time Hypothesis, there is no algorithm for $1|\operatorname{chains}(k),r_j,d_j|*$ with running time $f(k)n^{o(k)}$ for any computable function $f$.
    \end{enumerate}
\end{corollary}

\subsection{Hardness with prec-consistency}
Prec-consistency of time windows has allowed \FPT algorithms for $1|\operatorname{prec},r_j,d_j|C_{\max}$  parameterized by proper level \citep{MallemHM24}. In \Cref{construction:w2hardness}, the release times and deadlines of the chains are not consistent with the precedence constraints, thus the complexity of our problem for parameters $\#p,\#(d_j-r_j)$ assuming prec-consistency was an open question. By a slight modification of \Cref{construction:w2hardness}, we can obtain hardness even when the time windows are $\operatorname{prec}$-consistent, although we increase the value of the parameter $p_{\max}$.

\begin{construction}\label{construction:w2hardness_preccon}
    Let $p,q$ be any two distinct positive integers, without loss of generality, let $p<q$.
    We begin with an instance $(u_1,u_2,\ldots,u_\ell,v)$ of \textsc{Binary Shuffle Product} over the binary alphabet $\Sigma=\{p,q\}$ and follow \Cref{construction:w2hardness}. Now, instead of using the smaller value for the processing time of the guard jobs, we use the value $q$ for the processing time and time window of the guard jobs. In other words, $r_{g_i}=q\cdot i + \sum_{t=1}^iv[t]$ for every $i\in \{0,1,\ldots,n\}$ and $d_{g_i}=r_{g_i}+q$ for $i\in \{0,1,\ldots, n - 1\}$ (we define $d_{g_n}$ later). In order to have the release times and deadlines of the remaining $x$-jobs prec-consistent, we adjust the release times as follows. We let $r_{x_j^i}=(j+1)\cdot q + \sum_{t=1}^jv[t]$ and $d_{x_j^i}=r_{x_j^i}+\Delta$ where $\Delta = (|v|+1)\cdot q + \sum_{k=1}^{|v|}v[k]$. Finally, for $i=n$, we set $d_{g_i}=r_{g_i}+\Delta$. For all the guard jobs, the processing time is equal to $p_{g_i}=d_{g_i}-r_{g_i}$.
\end{construction}

Observe that the main idea of the construction does not change; now, even the deadlines of all jobs are greater than or equal to the deadlines in \Cref{construction:w2hardness}. Let $\mathcal{I}$ be an instance constructed from \Cref{construction:w2hardness} and $\mathcal{I}'$ an instance constructed from \Cref{construction:w2hardness_preccon}. Every feasible schedule of $\mathcal{I}$ is also a feasible schedule of $\mathcal{I}'$, on the other hand, the deadline and processing time of the last guard job $g_n$ prevents a feasible schedule to assign anything after $r_{g_n}$, hence it is also a feasible schedule for $\mathcal{I}$. Notice that whenever $x_j^i \prec x_{j+1}^i$, then 
\[r_{x_{j+1}^i}=  (j+2)\cdot q + \sum_{t=1}^{j+1}v[t] = \left((j+1)\cdot q + \sum_{t=1}^{j}v[t]\right)   + \left(v[j+1] + q\right)  \geq r_{x_j^i}+p_{x_j^i}.\] 
Symetrically, one can verify that $d_{x_j^i}\leq d_{x_{j+1}^{i}}-p_{x_{j+1}^{i}}$. This is an easy corollary of the fact that all $x$-jobs have the length of their time window equal to $\Delta$. Hence, the time windows are prec-consistent and we have $\#(d_j-r_j)\leq 2$. Note that the reduction still remains a pl-reduction. We thus obtain:

\begin{corollary}\label{cor:w2hardness_preccon}
    The problem $1|\operatorname{chains}(k),r_j,d_j|*$ is strongly {\sf XNLP}-hard parameterized by $k$ and strongly \textsf{NP}-hard even when $\#p\leq 3,\#(d_j-r_j)\leq 2$, and the time windows are prec-consistent.
\end{corollary}

\section{Algorithmic results}\label{sec:algorithms}

\subsection{\XP algorithm parameterized by the precedence width $w$}
We begin with an algorithm resolving a special case when the precedence relation consists of $k$ chains. Note that $k$ is in this case equal to the width $w$ of the precedence relation.

\begin{lemma}\label{lem:dp_xp_chains}
    The problem $1|\operatorname{chains}(k),r_j,d_j|C_{\max}$ can be solved in $n^{O(k)}$ time, where $n$ is the number of jobs.
\end{lemma}

\begin{proof}
    Let $J$ be the set of jobs to be scheduled, $|J|=n$ and let $J=Q_1\cup\cdots\cup Q_k$ be the partition of the set of jobs to the $k$ chains, i.e., $Q_i=\{x_1^i,x_2^i,\ldots,x_{|Q_i|}^i\},x_1^i\prec x_2^i\prec  \cdots \prec x_{|Q_i|}^i$. We use dynamic programming to solve the problem. For $j_i\in\{0,1,2,\ldots,|Q_i|\}$ we let $\operatorname{DP}(j_1,\ldots,j_k)$ be equal to the minimum makespan of a schedule that processes $j_i$ jobs from the $i$-th chain or $\infty$ if no such schedule exists. We will assume for simplicity that the earliest release time is equal to $0$, otherwise we shift all release times and deadlines accordingly. Clearly, we have $\operatorname{DP}(0,\ldots,0)=0$ and if $\operatorname{DP}(|Q_1|,\ldots,|Q_k|) \neq \infty$, it is the solution to our problem. On the other hand, if  $\operatorname{DP}(|Q_1|,\ldots,|Q_k|)= \infty$, there is no feasible solution. To maintain clarity and prevent subscript nesting, we adopt the notation $p(x)$, $r(x)$, and $d(x)$ for job $x$’s processing time, release time, and deadline.

    The DP transition is as follows. To compute $\operatorname{DP}(j_1,\ldots,j_k)$ for $j_i\in\{0,1,2,\ldots,|Q_i|\}$, note that one of the jobs in the chains for which $j_i\neq 0$ is scheduled last. For this particular job, we take the maximum of the previous schedule's makespan and the release time of that job. This is the earliest possible starting time $t_{i,j_1,\ldots,j_k}$ for the job $x^i_{j_i}$, i.e., $$t_{i,j_1,\ldots,j_k}=\max\left\{\operatorname{DP}(j_1,\ldots,j_i-1,\ldots,j_k),r(x^i_{j_i})\right\}.$$
    We now incorporate the deadlines. Note that if 
    $t_{i,j_1,\ldots,j_k}+p(x_{j_i}^i)>{d}(x_{j_i}^i)$, then there is no feasible schedule where the job $x_{j_i}^i$ is scheduled last. We thus obtain the following recurrence for $\operatorname{DP}(j_1,\ldots,j_k)$:
    \begin{align}
    \operatorname{DP}(j_1,\ldots,j_k)=\min_{i\in[k],j_i\neq 0}
    \begin{cases}
            t_{i,j_1,\ldots,j_k}+p(x_{j_i}^i) & \text{ if $t_{i,j_1,\ldots,j_k}+p(x_{j_i}^i)\leq {d}(x_{j_i}^i)$} \\
            \infty & \text{otherwise.}
    \end{cases}\label{eq:dp}
    \end{align}
    \begin{claim}\label{lem_dp_onedir} 
        If $\operatorname{DP}(j_1,\ldots,j_k)=t<\infty$, then there exist a schedule of makespan $t$ scheduling $j_i$ jobs from the $i$-th chain.
    \end{claim}
    \begin{proof}
        We proceed by induction on $j_1+j_2+\ldots+j_k$. Clearly, if $j_1=j_2=\cdots=j_k=0$, then the claim holds. Suppose that $\operatorname{DP}(j_1,\ldots,j_k)=t<\infty$. Then \[t=\max\{\operatorname{DP}(j_1,\ldots,j_i-1,\ldots,j_k),r(x_{j_i}^i)\}+p(x_{j_i}^i)\] for some $i\in[k],j_i\neq 0$ by (\ref{eq:dp}). By induction hypothesis, there is a schedule $\sigma$ for the smaller instance with $j_1,\ldots,j_i-1,\ldots,j_k$ jobs in the chains. We let $\sigma(x_{j_i}^i)=t_{i,j_1,\ldots,j_k}$. It is not hard to verify that $\sigma$ is indeed a feasible schedule of makespan $t$.
    \end{proof}
    
    \begin{claim}\label{lem_dp_twodir}
        If there exists a feasible schedule of makespan $t$ scheduling $j_i$ jobs from the $i$-th chain, then $\operatorname{DP}(j_1,\ldots,j_k)\leq t$.
    \end{claim}
    \begin{proof}
        We proceed by induction on $j_1+\cdots+j_k$. Clearly for $j_1=\cdots=j_k=0$, any schedule has makespan at least $0$ and $\operatorname{DP}(0,\ldots,0)=0$ by definition. Suppose now that $\sigma$ is a schedule of makespan $t$ scheduling $j_i$ jobs from the $i$-th chain. Note that $\sigma$ is feasible and respects the order $\prec$ on the jobs. Consider the last job scheduled by $\sigma$. This has to be one of $x_{j_{i^*}}^{i^*}$ for some $i^*\in[k], j_{i^*}\neq 0$. As the schedule has makespan $t$, the schedule $\sigma'$ restricted to the jobs without $x_{j_{i^*}}^{i^*}$ has makespan at most $t-p(x_{j_{i^*}}^{i^*})$. By induction hypothesis, we obtain $\operatorname{DP}(j_1,\ldots,j_{i^*}-1,\ldots,j_k)\leq t-p(x_{j_{i^*}}^{i^*})$. Note also that the release time of $x_{j_{i^*}}^{i^*}$ satisfies $r(x_{j_{i^*}}^{i^*})\leq t - p(x_{j_{i^*}}^{i^*})$. Hence $t_{i^*,j_1,\ldots,j_k}=\max\{\operatorname{DP}(j_1,\ldots,j_{i^*}-1,\ldots,j_k),r(x_{j_{i^*}}^{i^*})\}\leq t-p(x_{j_{i^*}}^{i^*})$, because both terms inside the maximum are bounded by the right hand side. It follows that the overall minimum in (\ref{eq:dp}) is at most $t-p(x_{j_{i^*}}^{i^*})+p(x_{j_{i^*}}^{i^*})=t$, hence $\operatorname{DP}(j_1,\ldots,j_k)\leq t$ and the proof is finished.
    \end{proof}

    By combining \Cref{lem_dp_onedir,lem_dp_twodir} we have that $\operatorname{DP}(|Q_1|,\ldots,|Q_k|)$ is the minimum makespan of a schedule if one exists, otherwise it is equal to $\infty$. Note that the number of entries in the dynamic programming table is $\prod_{j=1}^k(|Q_j|+1)\leq n^k$ where $|Q_j|\leq n$ and each transition can be computed in $O(k)$ time. This gives the final complexity bound $n^{O(k)}$, as promised.
\end{proof}

We are now ready to extend the dynamic programming to general precedence constraints of bounded width $w$. 

\begin{theorem}\label{thm:xp_width}
    The problem $1|\operatorname{prec},r_j,d_j|C_{\max}$ can be solved in $n^{O(w)}$ time, where $n$ is the number of jobs and $w$ is the width of the precedence relation.
\end{theorem}
\begin{proof}
    By Dilworth's theorem~\citep{Dilworth50}, the precedence relation can be decomposed into $w$ chains. The algorithm follows the same dynamic programming approach as in \Cref{lem:dp_xp_chains}, however when computing $\operatorname{DP}(j_1,\ldots,j_k)$, instead of computing the minimum over all $i\in[k],j_i\neq 0$, we consider only those $i\in[k],j_i\neq 0$ such that there is no job $x_j^{i'}$ for some $j\leq j_{i'}$ and $i'\neq i$ such that $x_{j_i}^i \prec x_j^{i'}$. The remaining reasoning and the proofs of the dynamic programming remain the same.
\end{proof}

Note that the running time from \Cref{thm:xp_width} is essentially optimal under the Exponential Time Hypothesis. Assuming the slice-wise polynomial space conjecture~(\Cref{conjecture:xnlp}), the space complexity $n^{O(w)}$ cannot be improved to $f(w)\operatorname{poly}(n)$ while keeping the time complexity $n^{f(w)}$ for some computable function $f$.

\subsection{{\sf XNLP}-completeness}

In this section, we discuss the membership of the problem $1|\operatorname{prec},r_j,d_j|C_{\max}$ parameterized by $w$ in {\sf XNLP}. Note that we could employ a similar idea as in the proof of {\sf XNLP} membership of the problem $Pm|\operatorname{prec},p_j=1|C_{\max}$ parameterized by $m+w$ in~\cite[Theorem 20]{BodlaenderGNS24_complete_xnlp}. The idea is to essentially use the dynamic programming approach given in \Cref{lem:dp_xp_chains}, but instead of computing the whole dynamic programming table, we nondeterministically guess the next transition. One issue is that we can no longer afford to decompose the precedence relation into $w$ chains beforehand. Although Dilworth's theorem provides a polynomial-time algorithm, it might be difficult to accomplish the same task with $f(w)\log n$ space complexity. We can overcome this instead by keeping track of the border between the set of already scheduled jobs and the remaining ones.

The crucial issue why this cannot be applied directly is the tracking of the current timestamp $t$, which represents the earliest time when the next job to be scheduled can start. This corresponds to the number $t_{i,j_1,j_2,\ldots,j_k}$ from the dynamic programming algorithm from \Cref{lem:dp_xp_chains}. Note that the input instance bit size is $n\log d_{\max}$ so a {\sf XNLP} algorithm would have to use at most $f(w)\cdot (\log n + \log \log d_{\max})$ bits, but storing $t$ itself may require $\Omega(\log d_{\max})$ bits. To overcome this difficulty, we propose a reduction rule that will reduce the maximum deadline to a function of the maximum processing time and number of jobs. Ultimately, we show that if the maximum processing time is bounded by $n^{f(w)}$ for some computable function, we can afford to store the timestamp $t$ explicitly.

\subparagraph{Reducing the maximum deadline}

In a general instance of $1|\operatorname{prec},r_j,d_j|*$, the deadlines and release times can be arbitrarily large. If we assume that the first release time is $0$, then an upper bound for the value of any release or deadline is certainly the maximum deadline $d_{\max}$. We show in \Cref{lem:reducing_d_max} that we can always bound $d_{\max}$ by $2n^2p_{\max}$, where $n$ is the number of jobs. We then argue that the reduction of $d_{\max}$ can be realized as a logspace reduction (\Cref{lem:red_logspace}), i.e., every bit of the resulting instance can be produced using at most $O(\log n+\log \log d_{\max})$ auxiliary memory. In particular, it is a pl-reduction, since we do not use the parameter. We then show that the problem $1|\operatorname{prec},r_j,d_j|*$ parameterized by $w$ is in {\sf XNLP} assuming that $p_{\max}$ is bounded by $n^{f(w)}$ for some (computable) function $f$~(\Cref{thm:xnlp_completeness_width}). Notice that this assumption is very likely in the scheduling applications where the processing times are often small.

\begin{lemma}\label{lem:reducing_d_max}
    Let $\mathcal{I}$ be an instance of $1|\operatorname{prec},r_j,d_j|*$ with $n$ jobs. There is an equivalent instance $\widehat{\mathcal{I}}$ with $d_{\max}\leq 2n^2p_{\max}$.
\end{lemma}

\begin{proof}
     Let $(u_\alpha)_{1\le\alpha\le 2n}$ be the ordered sequence of release times and deadlines. 
         If $u_1>0$ this means that no job will ever start before $u_1$, so we can just shift every date to the left  by $u_1$ in the instance and in any feasible schedule.
     
     If $u_{\alpha+1}-u_{\alpha}\leq \sum_{j=1}^np_j$ for all $\alpha\in[2n-1]$, then the instance satisfies $d_{\max}\leq 2n\cdot \sum_{i=1}^np_i\leq 2n^2p_{\max}$. Otherwise we proceed using a reduction rule. Suppose there exists $\alpha$ such that $u_{\alpha+1}-u_\alpha>\sum_{j=1}^np_j$. We compute $\gamma=u_{\alpha+1}-u_\alpha-\sum_{j=1}^np_j$. We build an instance $\mathcal{I}'$ from $\mathcal{I}$ by applying an offset $-\gamma$ to each $r_i,d_i\ge u_{\alpha+1}$. The remaining $r_i,d_i\leq u_{\alpha}$ remain intact.

    Assume that there is a feasible schedule $\sigma$ for $\mathcal{I}$. Between $u_\alpha$ and $u_{\alpha+1}$ as there are no release times, we can assume that there is no idle time between two jobs started in this interval. This implies that $u_{\alpha+1}$ is then strictly greater than the completion time of each job started in this interval. And the next jobs start after or at $u_{\alpha+1}$. So we can build a feasible schedule of $\mathcal{I}'$ by applying an offset $-\gamma$ to each $\sigma(j)\ge u_{\alpha+1}$. The precedence constraints are fulfilled. Two cases occurs: if $r_j\le u_\alpha$ then $\sigma'(j)\ge u_{\alpha+1}-\gamma\ge r_j$, and $\sigma'(j)=\sigma(j)-\gamma\le d_j-\gamma=d'_j$. Otherwise the offset $-\gamma$ is applied to the time window of $j$ and the schedule.
    
    Conversely, a feasible schedule of $\mathcal{I}'$ will lead to a feasible schedule for $\mathcal{I}$ by applying an offset $\gamma$ to all starting times not less than $u_{\alpha+1}-\gamma$.
    
    Doing this transformation iteratively gives the instance $\widehat{\mathcal{I}}$.
    Notice that the value of $\widehat{u_\alpha}$ in the instance $\widehat{\mathcal{I}}$ can then be expressed as follows:
\begin{equation}
    \widehat{u_\alpha}=\min(u_1,0)+ \sum_{\beta=1}^{\alpha-1}\min\left\{\sum_{j=1}^n p_j,u_{\beta+1}-u_\beta\right\}\label{eq:u_alpha}
\end{equation}
\end{proof}

\begin{lemma}\label{lem:red_logspace}
    The computation of the instance $\widehat{\mathcal{I}}$ from \Cref{lem:reducing_d_max} can be realized as a logspace reduction.
\end{lemma}
\begin{proof}
    We demonstrate that the transformed values
    $\widehat{u_{\alpha}}$ can be computed using a logspace uniform $\textsf{FNC}^1$ circuit family. The desired claim then follows from \Cref{lem:ftc0_is_logspace_red}. We construct a circuit that evaluates \Cref{eq:u_alpha}. The input consists of (the bits of) the $3n$ numbers $p_1,r_1,d_1,\ldots,p_n,r_n,d_n$ and the output corresponds to the (bits of) transformed numbers $\widehat{p_1},\widehat{r_1},\widehat{d_1},\ldots \widehat{p_n},\widehat{r_n},\widehat{d_n}$. The number of input bits is thus at most $3n\cdot \lceil \log_2(d_{\max}+1)\rceil$. We decompose the circuit into constant number of ${\sf FNC}^1$ circuits which we now describe. For a visual representation of the whole construction refer to \Cref{fig:circuit_reduction}.

    \begin{description}
        \item[Global Sum]: The term $S=\sum_{j=1}^np_j$ can be computed by a ${\sf FNC}^1$ circuit by~\Cref{lem:addition_of_n_numbers}.
        \item[Pairwise Comparison and Sorting:] Let $a_1,a_2,\ldots,a_{2n}=r_1,d_1,r_2,d_2,\ldots,r_n,d_n$. For every pair $(i,j)$ we compute a comparison bit $x_{ij}$ as the indicator of $a_i<a_j\vee (a_i=a_j\wedge i<j)$. This can be accomplished in $\textsf{FNC}^1$ by subtracting $a_i-a_j$ and examining the sign bit.
        
        The \emph{rank} of each element is $R_j=\sum_{i=1}^{2n}x_{ij}$. This represents the position of $a_j$ in the sorted sequence $u_1,u_2,\ldots,u_{2n}$. Computing these ranks is another instance of addition which can be done using~\Cref{lem:addition_of_n_numbers}. Output of this layer are the (bits of) sorted values $u_1,u_2,\ldots,u_{2n}$ where $u_k$ is extracted via a multiplexer network: 
        \[
        u_k=\bigvee_{i=1}^{2n}\Big(a_i\wedge \mathbb{I}[R_i=k]\Big).
        \]
        Here, the indicator function $\mathbb{I}[R_i=k]$ evaluates to a single bit, which is then fanned out into a bitmask of ones of the same length as $a_i$ to either pass the value of mask it to zero. Note that the $\bigvee$ over $n$ is expanded to a binary tree of depth $O(\log n)$ using only or-gates with fan-in $2$.

        \item[Computing gaps:] Given the output of the previous layer $u_1,u_2,\ldots,u_{2n}$, we compute the pairwise differences $\Delta_{\beta}:=u_{\beta+1}-u_{\beta}$. This only uses the circuit from~\Cref{lem:addition_of_n_numbers} for two numbers.

        \item[Interval Clamping and Prefix Sum:] We need to compute $M_\beta :=\min(S,\Delta u_{\beta})$. To achieve this, we evaluate $S<\Delta_{\beta}$ and then use a multiplexer with $1$ selection bit (or $2$-to-$1$ multiplexer) to either forward $S$ or $\Delta u_{\beta}$. The updated coordinates are generated using a parallel prefix sum:
        \[
            \widehat{u_\alpha}=\min(u_1,0)+\sum_{\beta=1}^{\alpha-1}M_\beta.
        \]
        Because iterated addition is in $\textsf{FNC}^1$ (\Cref{lem:addition_of_n_numbers}), all $\widehat{u_\alpha}$ values can be computed simultaneously.
        
        \item[Output remapping] We have computed the compressed sorted sequence $\widehat{u_1},\ldots,\widehat{u_{2n}}$. It remains to reorder them to their original positions to obtain $\widehat{a_1},\widehat{a_2},\ldots,\widehat{a_{2n}}=\widehat{d_1},\widehat{r_1},\widehat{r_2},\widehat{d_2},\ldots,\widehat{r_n},\widehat{d_n}$. Since the transformation $u\mapsto \widehat{u}$ is nondecreasing, the relative order of the elements does not change. Therefore, we reuse the ranks $R_j$. The final values are routed using similar multiplexer logic as before:
        \[
        \widehat{a_j} = \bigvee_{k=1}^{2n} \Big(\widehat{u_k} \wedge \mathbb{I}[R_j = k]\Big)
        \]
    \end{description}
    The whole construction clearly yields a $\FNC^1$ circuit and it is not hard to observe that the circuit family is logspace uniform.
\end{proof}

\begin{figure}[htpb]
    \centering
    \begin{tikzpicture}[
        font=\sffamily\small,
        >=Stealth,
        block/.style={rectangle, draw=black, thick, fill=blue!5, text width=4.2cm, align=center, minimum height=1.2cm, rounded corners=2pt},
        smallblock/.style={rectangle, draw=black, thick, fill=blue!5, text width=3.8cm, align=center, minimum height=1.2cm, rounded corners=2pt},
        data/.style={rectangle, draw=black!60, fill=gray!10, text width=3.5cm, align=center, minimum height=0.6cm},
        bus/.style={draw, ->, thick}
    ]

        \node[data, text width=8.5cm] (InSingle) at (0, 7.2) {\textbf{Input Instance} ($3n$ numbers)\\$p_1, r_1, d_1, \ldots, p_n, r_n, d_n$};

        \node[smallblock] (SumP) at (-3.0, 4) {\textbf{Global Sum}\\$S = \sum_{j=1}^n p_j$};
        \node[smallblock] (RankCalc) at (3.0, 4) {\textbf{Comparison \& Rank}\\$x_{ij} = \mathbb{I}[a_i < a_j \lor (a_i = a_j \land i < j)]$\\$R_j = \sum_{i=1}^{2n} x_{ij}$};
        
        \node[smallblock] (SortMux) at (3.0, 1.5) {\textbf{Sorting Multiplexer}\\$u_k = \bigvee_{i=1}^{2n} \Big(a_i \wedge \mathbb{I}[R_i=k]\Big)$};

        \node[block] (GapClamp) at (0, -0.5) {\textbf{Gap \& Clamping}\\$M_\beta = \min(S, u_{\beta+1} - u_\beta)$};

        \node[block] (PrefixSum) at (0, -2.5) {\textbf{Parallel Prefix Sum}\\$\widehat{u_\alpha} = \min(u_1,0) + \sum_{\beta=1}^{\alpha-1} M_\beta$};

        \node[block] (OutRoute) at (0, -4.5) {\textbf{Output Remapping}\\$\widehat{a_j} = \bigvee_{k=1}^{2n} \Big(\widehat{u_k} \wedge \mathbb{I}[R_j=k]\Big)$};

        \node[data, text width=8.5cm] (OutFinal) at (0, -6.3) {\textbf{Reduced Output Instance}\\$p_1, \widehat{r_1}, \widehat{d_1}, \ldots, p_n, \widehat{r_n}, \widehat{d_n}$};

        
        \draw[bus] ($(InSingle.south) - (3.0, 0)$) -- node[midway, left] {$p_1, \ldots, p_n$} (SumP.north);
        \draw[bus] ($(InSingle.south) + (3.0, 0)$) -- node[midway, right] {$a_1, \ldots, a_{2n}$} (RankCalc.north);
        
        \draw[bus] (-3.0, 5.4) -- (-5.5, 5.4) -- node[midway, left, align=right] {$p_1, \ldots, p_n$} (-5.5, -6.3) -- (OutFinal.west);

        \draw[bus] (3.0, 5.4) -- (0.0, 5.4) -- node[midway, left, align=right] {$a_i$} (0.0, 1.5) -- (SortMux.west);
        
        \draw[bus] (SumP.south) -- node[midway, left] {$S$} (-3.0, -0.3) -- ([yshift=2mm]GapClamp.west);
        
        \draw[bus] (RankCalc) -- (SortMux) node[midway, right] {$R_j$};
        
        \draw[bus] (RankCalc.east) -- (5.5, 4.0) -- node[midway, right, align=left] {$R_j$} (5.5, -4.5) -- (OutRoute.east);
        
        \draw[bus] (SortMux.south) -- (3.0, -2.7) -- ([yshift=-2mm]PrefixSum.east) node[pos=0.7, below] {$u_1$};
        \draw[bus] (3.0, -0.3) -- ([yshift=2mm]GapClamp.east) node[midway, above] {$u_k$};

        \draw[bus] (GapClamp) -- (PrefixSum) node[midway, right] {$M_\beta$};
        \draw[bus] (PrefixSum) -- (OutRoute) node[midway, right] {$\widehat{u_k}$};
        \draw[bus] (OutRoute) -- (OutFinal);

    \end{tikzpicture}
    \caption{Data flow of the $\textsf{FNC}^1$ circuit computing the instance transformation. Solid black lines represent multi-bit data buses. The processing times $p_j$ are forwarded directly to the output, while the time coordinates $a_i$ are sorted, clamped, prefix-summed, and remapped using ranks $R_j$ to form the reduced coordinates.}
    \label{fig:circuit_reduction}
\end{figure}

We can now establish that the single machine problem with precedence constraints and time windows is {\sf XNLP}-complete, provided that the maximum processing time is bounded by $n^{f(w)}$.

\begin{theorem}\label{thm:xnlp_completeness_width}
    Problem $1|\operatorname{prec},r_j,d_j|*$ is {\sf XNLP}-complete parameterized by the width $w$ of the precedence relation when the maximum processing time satisfies $p_{\max}\le n^{f(w)}$ for some computable function $f$.
\end{theorem}
\begin{proof}
The hardness part follows from \Cref{thm:w2hardness} even when the numbers are encoded in unary. We establish the {\sf XNLP} membership by providing a non-deterministic algorithm that uses $g(w)\log n$ auxiliary memory and runs in $g(w)\operatorname{poly}(n)$ time for some computable function $g$ (in fact, the running time will be just polynomial in $n$).

We demonstrate the algorithm for the case when $d_{\max}\leq 2n^2p_{\max}$. Note that since $1|\operatorname{prec},r_j,d_j|*$ is logspace-reducible (hence pl-reducible) to $1|\operatorname{prec},r_j,d_j,d_{\max}\leq 2n^2p_{\max}|*$ by \Cref{lem:red_logspace}, showing \XNLP memberhsip for $1|\operatorname{prec},r_j,d_j,d_{\max}\leq 2n^2p_{\max}|*$ shows \XNLP memberhsip for $1|\operatorname{prec},r_j,d_j|*$ by \Cref{lem:xnlp_pl_reductions}. Note that the processing times do not change by application of the transformation from \Cref{lem:reducing_d_max}.

We have $d_{\max}\leq 2n^2p_{\max}\leq n^{f(w)+3}$. Assume that a partial schedule of a set of jobs $S$ has been built. We know that if $j\in S$, all its predecessors are in $S$. It is not hard to see that $S$ is uniquely determined by the set $M\subseteq S$ of its maximal elements. Given the set $M$ of maximal elements, $S$ is given by the set of jobs $j$ such that there is a path from $j$ to some element of $M$. We call $M\subseteq S$ the \emph{border set} of $S$. Note that $M$ is an independent set with respect to the precedence relation, hence $|M|\leq w$. While non-deterministically building the schedule, we store the set $M$ using at most $O(w\log n)$ bits and the current timestamp $t$, which represents the earliest time when the next job can be scheduled. The current timestamp $t$ is updated in a similar fashion as in the dynamic programming from \Cref{lem:dp_xp_chains}. Notice that $t\leq d_{\max}$ hence we can store it using at most $O(\log d_{\max})\leq O((f(w)+3)\log n)$ bits.

The algorithm works as follows. We store $M$ together with the timestamp $t$, initializing $M=\emptyset$ and $t=0$ and a counter $c$ which counts the number of currently scheduled jobs. In each step, we nondeterministically guess the next job $i$ to be scheduled. 
We then check whether $i$ can be scheduled at timestamp $t$. To do this, we need to verify three properties for $i$:
\begin{enumerate}[label=\alph*)]
    \item verify that $i$ is not already scheduled,\label{prop:a}
    \item verify that all predecessors of $i$ are already scheduled,\label{prop:b}
    \item verify that $i$ can be scheduled at time $\max\{r_i,t\}$.\label{prop:c}
\end{enumerate}
We show that all of the properties (\ref{prop:a})-(\ref{prop:c}) can be verified by a {\sf NL} algorithm.

(\ref{prop:a}): $i$ is not already scheduled if there is no path from $i$ to $M$. Notice that this problem is the complement of \textsc{STcon+}, hence in {\sf coNL} by \Cref{obs:stcon_plus} and since ${\sf coNL}={\sf NL}$ by \Cref{cor:immerman_szelepczenyi}, it is in {\sf NL}.

(\ref{prop:b}): We need to check that every direct predecessor $j$ of $i$ is in $S$. That is, for every predecessor $j$ of $i$ there is a path from $j$ to $M$ (recall that this is the characterization of $S$). We demonstrate that the complementary problem is in {\sf NL}. That is: "Does there exist a predecessor $j$ of $i$ such that there is no path from $j$ to $M$?". Formally $L=\{(G,i,M)\mid \exists (j,i)\in E(G):\forall m\in M:\text{there is no path from $j$ to $m$\}}$. We solve $L$ by nondetereministically guessing $j$ and we are then left with solving \textsc{STcon+} which is in {\sf NL} by \Cref{obs:stcon_plus}. We again invoke \Cref{cor:immerman_szelepczenyi} and this shows that deciding whether every direct predecessor $j$ of $i$ is in $S$ can be decided using a {\sf NL} algorithm.

(\ref{prop:c}): Lastly, we check whether $i$ fits into the current timestamp. The minimum time when $i$ can be scheduled is $t'=\max\{t,r_i\}$. We then need to verify whether $t'+p_i\leq d_i$. This is true if and only if $t+p_i\leq d_i \wedge r_i+p_i\leq d_i$. All this can be accomplished in logspace using \Cref{lem:addition_of_n_numbers}.

If any of the above checks (\ref{prop:a})-(\ref{prop:c}) fail, we reject the current branch (note that since everything is done by a {\sf NL} algorithm, if a check fails, all nondetereministic branches of the current one reject). Otherwise, we schedule $i$, we increment $c$ by one, update $t$ to $t'+p_i$ and we update the border set $M$ to reflect the currently scheduled set of jobs $S$. We remove from $M$ all predecessors of $i$ (if any of them are present) and add $i$ to $M$. We accept if all $n$ jobs are scheduled, i.e., $c=n$, otherwise we return back to the nondeterministic guess of the next job to be scheduled. This finishes the description of the nondeterministic algorithm.

The time complexity is clearly polynomial. It remains to argue about the space complexity. We store the border set $M$ using at most $O(w\log n)$ bits and we are using $O((f(w)+3)\log n)$ bits for the timestamp $t$ and we need the same amount of bits for its manipulation (comparison and addition with $r_i,d_i,p_i$). Hence the algorithm is a {\sf XNLP} algorithm and this finishes the proof.
\end{proof}

\begin{corollary}
   The problem $1| \operatorname{chains}(k), r_j , d_j |*$ is  {\sf XNLP}-complete parameterized by $k$, provided that $p_{max}\le n^{f(k)}$ for some computable function $f$.
\end{corollary}

\begin{corollary}
    The problem $1|\operatorname{prec},r_j,d_j|*$ is {\sf XNLP}-complete parameterized by the width $w$ of the precedence relation given that all numbers in the instance are encoded in unary.
\end{corollary}
\begin{proof}
    Note that the hardness given in \Cref{thm:w2hardness} produces an instance where everything may be encoded in unary. For the membership part, the same algorithm as in the proof of \Cref{thm:xnlp_completeness_width} applies. Note that the bit-size of the instance is now $n\cdot d_{\max}$ (as opposed to $n\cdot \log d_{\max}$ if everything was encoded in binary) and we are allowed to have space complexity up to $g(w)\cdot (\log n+\log d_{\max})$.
\end{proof}

\subsection{Polynomial time algorithm when $\#(d_j-r_j)=1$}
\begin{theorem}\label{thm:polytime_onewindow}
    If $\#(d_j-r_j)=1$, then $1|\operatorname{prec},r_j,d_j|C_{\max}$ with $\operatorname{prec}$-consistent time windows is solvable in polynomial time.
\end{theorem}
\begin{proof}
    We use the notion of \emph{proper level} introduced by \cite{MallemHM24}. The \emph{proper level} of a job $j\in J$ is the number of jobs $j'$ with $r_{j'}<r_j$ and $d_j<d_{j'}$. In other words, it is the number of jobs whose time window strictly contains the time window of $j$. The proper level of the entire instance is the maximum proper level of a job.
    \cite{MallemHM24} showed that $1|\operatorname{prec},r_j,d_j|C_{\max}$ is \FPT parameterized by the propert level. Note that when $\#(d_j-r_j)=1$, then the proper level is equal to $0$, hence the problem is solvable in polynomial time in this case.
\end{proof}

\subsection{\FPT algorithm when chains are uniform}

We investigate the complexity of $1|\operatorname{chains}(k),p_Q,r_Q,{d}_Q|*$. In other words, we assume that in each chain, all jobs have the same processing time, release time, and deadline. Observe that in such a scenario, we can drop the constraints on the precedence between jobs because in every schedule we can reorder them according to the original requirement $\prec$, hence this can be reduced to the \FPT algorithm by \cite[Theorem 27]{KnopKLMO19} parameterized by $\#(r_j,d_j,p_j)$. We have shown:

\begin{theorem}\label{thm:knop_red}
    The problem $1|\operatorname{chains}(k),p_j=p_Q,r_Q,{d}_Q|*$ is \FPT parameterized by $k$.
\end{theorem}

Note that in \Cref{thm:w2hardness}, all the items have only two possible processing times, although the processing times are not chain-uniform.

\section{Strategies for designing RPA scheduling algorithms}
\label{sec:rpa_pc_relation_discussion}

From the perspective of designing scheduling algorithms, the RPA domain is relatively unique. At first glance, it may appear similar to problems related to compute clusters~\citep{Khallouli2022}. However, clusters are often used by a wider group of users and the spectrum of computational jobs is much more diverse. Conversely, RPA is connected with the business processes of one company which are well defined and under the company's control. Nevertheless, the common objective is efficient utilization of computational resources while jobs needs to be finished on time. Additionally, the computational time available for scheduling is limited, as the scheduling algorithm needs to respond to changing demand in the form of new items generated by the RPA.

Scheduling of RPA has several specifics that may indicate that the related scheduling problems may be easier to solve than, for example, production scheduling problems. In RPA one may assume that items from a given queue have equal processing times and that there is a constant difference between $r_j$ and $d_j$ for all items or all items of a queue (i.e., SLA of the corresponding business process). Nevertheless, Table~\ref{tab:table_results} and the new complexity results presented in this paper show that even single machine variant of the problem can be difficult to solve. An example is the problem $1|\operatorname{chains},r_j,d_j,p_j\in\{p,q\}|*$, that was shown to be \pNP-hard for parameter $\#(d_j-r_j)$ in Theorem~\ref{thm:w2hardness}. However, there are some problem properties that can be exploited for designing sufficiently fast exact or highly efficient heuristic scheduling algorithms.

\subsection{Queues and $\#(d_j-r_j)=1$}

RPA scheduling instances typically consist of hundreds but more likely thousands of items. Therefore, if the problem is solved via, e.g., constraint programming or a heuristic, it is meaningful to assume coarse-grained models. One natural approach is to group multiple items from the same queue into a single job. Since each queue in RPA represents a specific business process, all items within a queue are of the same type and are processed by the same worker or resource. Therefore, grouping items from the same queue does not combine incompatible job types. A clever strategy for merging items into jobs is essential for achieving the desired computation time of the algorithm while preserving the solution quality and feasibility.

In each queue, there is a time stamp at the arrival of an item and then the deadline is fixed with respect to this time stamp. Therefore, the order in the queue is determined by the time stamp and we can assume that in a queue all release times are different.
Notice that if there are at most $R$ jobs that can arrive at the same time $t$, we can choose $\varepsilon<\frac{1}{R+1}$ and set the release times to $t,t+\varepsilon,\ldots, t+R\varepsilon$ and the deadlines accordingly. Any integer schedule that respects the deadline constraints will also satisfy the initial deadlines.

Assuming different release times in each queue, we know that if $\#(d_j-r_j)=1$ the problem with independent jobs has proper level $0$ and that the EDD (Earliest Due Date first) schedule is dominant for the maximum latency $L_{\max}=\max_j(C_j-d_j)$. This implies that the EDD schedule will satisfy the order of items in each queue and therefore the items can be executed in the EDD order within each job.

It is interesting to note that problems $P|r_j,p_j=p,d_j|*$ and $P|r_j,p_j=p,d_j|\sum C_j$ (regardless $\#(d_j-r_j)=1$ or not) can be also solved in polynomial time using linear programming~\citep{Brucker2008}. Therefore, if jobs are constructed from items such that their processing times are practically identical and there is no need for precedence constraints (e.g., when all items arrive at the same time), then the problem can be solved in polynomial time.

\subsection{Equal processing time per queue and proper intervals in each queue}
Now, let us focus on the properties of queues in RPA, i.e., for each queue $Q$ there is a specific $d_j-r_j= d_Q-r_Q$ and a specific processing time $p_Q$. Assuming similarly that release times are different in each queue, we can observe that the independent items setting can have large proper levels (and even large path-width) so it seems unlikely to get a direct outcome of \FPT results with those parameters (see~\Cref{tab:table_results}).

However, we can also say that as items in a queue will always be performed in the non-decreasing deadline order (since a swap of items in a different order is always possible) on a single machine as long as they share the same processing time and that $r_i\le r_j$ implies $d_i\le d_j$ in a queue. So that we can neglect the precedence constraint in this case.
By the result of \cite{CARLIER2004}, in this case EDD will produce a schedule with $L_{\max}^{EDD}\le L_{\max}^{OPT}+2(p_{\max}-1)$ even if we have parallel machines.
The result from \cite{CARLIER2004} can be slightly improved for the single machine case to
\begin{equation}
    L_{\max}^{EDD}\le L_{\max}^{OPT} + p_{\max}-1.
\end{equation}
To see this, consider the EDD schedule, and let $i$ be a job such that $L_{\max}^{EDD}=C_i^{EDD}-d_i$. Let $t\le C_i-p_i$ be the rightmost time such that either there is an idle period in interval $[t-1,t)$ or there is a job $j$  such that $d_j>d_i$ and $C_j=t$. Let $J'$ be the set of jobs performed between $t$ and $C_i$. If $k\in J'$ then $r_k > t-p_j$ if $j$ exists or $r_k\ge t$ otherwise.

Now, $C_i^{EDD}=t+\sum_{k\in J'}p_k$. 
So, if $j$ exists, then
\begin{align}
 C_i^{EDD}-p_j=t+\sum_{k\in J'}p_k-p_j<\min_{k\in J'}r_k+\sum_{k\in J'}p_k\le\max_{k\in J'}C_k^{OPT}   
\end{align}
which implies 
\begin{align}
   L_{\max}^{OPT}\ge \max_{k\in J'}C_k^{OPT}-d_i> C_i^{EDD}-d_i-p_j=L_{\max}^{EDD}-p_j. 
\end{align}

In the case there is no such $j$, we have $C_i^{EDD}=\min_{k\in J'}r_k+\sum_{k\in J'}p_k\le\max_{k\in J'}C_k^{OPT}$, so that $L_{\max}^{EDD}\le L_{\max}^{OPT}$.

This property can be exploited in two-stage scheduling approaches, where in the first phase the jobs are assigned to machines and in the second one the assigned jobs are sequenced. The EDD rule can be used in both phases to estimate whether the assignment or sequencing will comply with the prescribed SLA. Furthermore, since the maximum error of the EDD rule is given by $p_{\max}$, this should be taken into account when deciding how many items to merge into one job, as discussed in the previous subsection. It is evident that larger jobs can result in a greater discrepancy between the EDD schedule and the optimal sequencing.

If the RPA scheduling problem is solved using constraint programming, the assignment computed in the first phase of the two-stage approach can be used to warm-start a constraint programming solver. The solver is provided with the assignment of jobs to machines, while their start times are left uninitialised.
A good warm start can guide the solver toward promising regions of the search space, thereby speeding up the search. This technique is seen as one of the most efficient approaches to speed up constraint programming model solving, e.g.,~\cite{kovacs_et_al:2021}.
If the constraint programming model remains inefficient, then the assignment of jobs can be fixed in the model.

\subsection{Problem decomposition exploiting fixed parameters}
As can be seen from Table~\ref{tab:table_results}, a fixed processing time, release and deadline per queue is advantageous for single-machine RPA-related scheduling. On the other hand, in general, even with just two fixed processing times, the problem is difficult to solve (see~\Cref{thm:w2hardness}). However, if the number of queues assigned to one machine is a small constant, then even this problem can be solved exactly via dynamic programming in polynomial time (see~\Cref{thm:xp_width}).
Conversely, problems $1|\operatorname{chains}(k),p_j=p_Q,r_Q,{d}_Q|*$ and $1|r_j,d_j|C_{\max}$ can be solved efficiently for a fixed $k$ and $\#(r_j,d_j,p_j)$ respectively. This suggests that such problems can be solved relatively quickly by the referenced algorithms for small values of these parameters. It is also reasonable to expect general constraint programming solvers to be highly efficient if the model is appropriately defined with regard to the fixed parameters. For example, the algorithm may construct jobs from items in such a way that $\#(r_j,d_j,p_j)$ is kept low. Another strategy may be to restrict the number of queues that can be assigned to single machine, i.e., restrict $k$. This may cut off the optimal solution but on the other hand it may eliminate many symmetrical solutions and make the sequencing problem easier to solve.

This strategy provides an opportunity to design, for example, a Logic-Based Benders decomposition approach, where the assignment of jobs to machines is carried out by a master problem formulated as integer linear programming. This assignment can be modeled as a bin packing, with the aim of balancing machine loads while also minimizing either $k$ or $\#(r_j,d_j,p_j)$. These parameters can be minimized by limiting the feasible assignment or by the specific construction of jobs from items. Even if this strategy eliminates some  optimal solutions, the approach can still be used as a highly efficient heuristic. Subproblems operating over a simplified solution space verify the feasibility of the assignment by scheduling the jobs on the machine and indicating infeasibility to the master problem by generating infeasibility cuts.

\section{Conclusion}\label{sec:conclusion}

We studied the complexity-theoretical side of scheduling problems related to RPA from the parameterized complexity perspective. With focus on parameters that can be considered small in real-world instances of RPA scheduling, we encountered \pNP-hardness results for most of them, namely $\#r_j$, $\#d_j$, $\#p_j$, $p_{\max}$, $\#(d_j-r_j)$ and strong {\sf XNLP}-hardness of $1|\operatorname{chains}(k),r_j,d_j|*$ parameterized by the number of chains $k$. On the positive side, we demonstrated that the problem becomes \FPT when the release times, deadlines, and processing times are chain-uniform and with general precedences, it becomes \XP parameterized by the width of the precedence relation. If the processing time is bounded $n^{f(w)}$ or encoded in unary, the problem becomes {\sf XNLP}-complete parameterized by $w$. Finally, we identified a tractable case when there is a single length of a time window, i.e., $\#(d_j-r_j)=1$, which corresponds to the practical instance of RPA, where all the jobs share a common deadline relative to their release time.

We propose the following future research directions. An open question is the complexity of $1|\operatorname{chains},r_j,d_j|C_{\max}$ parameterized by $\#(r_j,d_j,p_j)$. Note that the \FPT algorithm of \cite{KnopKLMO19} is not applicable in the presence of precedences, and by the result of \cite{LeslieSh92}, the problem with general precedences is \pNP-hard parameterized by $\#(r_j,d_j,p_j)$. 

Another open question is the complexity of $1|\operatorname{chains}(k),r_j,d_j,p_j=p_Q|*$. This scenario corresponds to the jobs in the same chain having uniform processing time, but not necessarily release or deadline. Is the problem \FPT parameterized by $k$?

Lastly, we ask if the problem $1|\operatorname{prec},r_j,d_j|*$ remains in the class {\sf XNLP} even if the processing times are encoded in binary and are unrestricted.

\section*{Acknowledgement}
This work was supported by the Technology Agency of the Czech Republic under the Project FW11020080 and co-funded by the European Union under the project ROBOPROX - Robotics and Advanced Industrial Production (reg. no. CZ.02.01.01/00/22\_008/0004590).

\bibliographystyle{elsarticle-harv}
\bibliography{references}

\end{document}